\documentclass[11pt,a4paper,reqno]{amsart}%
\usepackage{amsthm,amsmath,amsfonts,amssymb,amsxtra,appendix,bookmark,dsfont,latexsym}

\makeatletter
\usepackage{hyperref}
\usepackage{delarray,a4,color}
\usepackage{euscript}
\usepackage[latin1]{inputenc}

\theoremstyle{plain}
\newtheorem{theorem}{Theorem}[section]
\newtheorem{lemma}[theorem]{Lemma}
\newtheorem{corollary}[theorem]{Corollary}
\newtheorem{proposition}[theorem]{Proposition}

\theoremstyle{definition}

\newtheorem{example}[theorem]{Example}
\theoremstyle{remark}
\newtheorem{remark}[theorem]{Remark}

\numberwithin{equation}{section}

\DeclareMathOperator{\Tr}{Tr}
\DeclareMathOperator{\tr}{Tr}

\def\geqslant{\ge}
\def\leqslant{\le}
\def\bq{\begin{eqnarray}}
\def\eq{\end{eqnarray}}
\def\bqq{\begin{eqnarray*}}
\def\eqq{\end{eqnarray*}}
\def\nn{\nonumber}

\def\eps{\varepsilon}
\def\wto{\rightharpoonup}

\newcommand{\norm}[1]{\left\lVert #1 \right\rVert}

\newcommand\1{{\ensuremath {\mathds 1} }}

\newcommand{\brar}{\right|}
\newcommand{\bral}{\left\langle}
\newcommand{\ketr}{\right\rangle}
\newcommand{\ketl}{\left|}
\newcommand{\HartE}{e_{\rm H}}

\newcommand{\gammaP}{\gamma_{\Psi}}

\newcommand{\dM}{{\rm d}M}

\renewcommand{\epsilon}{\varepsilon}

\def\cF {\mathcal{F}}
\def\cB {\mathcal{B}}

\def\R {\mathbb{R}}
\def\C {\mathbb{C}}

\def\cS {\mathcal{S}}

\def\cP {\mathcal{P}}
\def\E {\mathcal{E}}
\def\cE {\mathcal{E}}

\def\cK{\mathcal{K}}

\def\R {\mathbb{R}}
\def\C {\mathbb{C}}

\def\gS{\mathfrak{S}}

\def\cS {\mathcal{S}}

\def\E {\mathcal{E}}
\def\cM {\mathcal{M}}

\def\gH{\mathfrak{H}}
\newcommand\ii{{\ensuremath {\infty}}}
\newcommand\pscal[1]{{\ensuremath{\left\langle #1 \right\rangle}}}

\renewcommand{\leq}{\leqslant}
\renewcommand{\geq}{\geqslant}

\newcommand{\eH}{\ensuremath{e_{\text{\textnormal{H}}}}}
\newcommand{\cEH}{\ensuremath{\cE_{\text{\textnormal{H}}}}}

\title[Derivation of Hartree's theory]{Derivation of Hartree's theory for generic mean-field Bose systems}

\author[M. Lewin]{Mathieu Lewin}
\address{CNRS \& Laboratoire de Math\'ematiques (UMR 8088), Universit\'e de Cergy-Pontoise, F-95000 Cergy-Pontoise, France.}
\email{mathieu.lewin@math.cnrs.fr}

\author[P.T. Nam]{Phan Th\`anh Nam}
\address{CNRS \& Laboratoire de Math\'ematiques (UMR 8088), Universit\'e de Cergy-Pontoise, F-95000 Cergy-Pontoise, France.}
\email{phan-thanh.nam@u-cergy.fr}

\author[N. Rougerie]{Nicolas Rougerie}
\address{Universit\'e Grenoble 1 \& CNRS,  LPMMC (UMR 5493), B.P. 166, F-38 042 Grenoble, France}
\email{nicolas.rougerie@grenoble.cnrs.fr}

\date{November 6, 2013}


\begin{document}

\begin{abstract}
In this paper we provide a novel strategy to prove the validity of Hartree's theory for the ground state energy of bosonic quantum systems in the mean-field regime. For the known case of trapped Bose gases, this can be shown using the strong quantum de Finetti theorem, which gives the structure of infinite hierarchies of $k$-particles density matrices. Here we deal with the case where some particles are allowed to escape to infinity, leading to a lack of compactness. Our approach is based on two ingredients: (1) a weak version of the quantum de Finetti theorem, and (2) geometric techniques for many-body systems. Our strategy does not rely on any special property of the interaction between the particles. In particular, our results cover those of Benguria-Lieb and Lieb-Yau for, respectively, bosonic atoms and boson stars.
\end{abstract}

\maketitle

\setcounter{tocdepth}{2}
\tableofcontents

\section{Introduction}

In this paper, we consider a quantum system composed of a very large number of interacting particles.  Because of the correspondingly large number of degrees of freedom, it is extremely hard to describe the precise behavior of a system of this kind. It is thus often useful to resort to approximate theories which are simpler to deal with. One of the major issues in many-body physics is then to justify the validity of the effective models, that is, to relate them to the many-body problem in a particular regime.

For bosons, the simplest effective theory can be obtained by assuming that all particles are in the same quantum state. This leads to the celebrated nonlinear model introduced by Hartree in~\cite{Hartree-28}.\footnote{More precisely, the theory was intended to be used for the electrons in an atom but it had to be corrected later by Fock since electrons are not bosons, but rather fermions.} In this theory, the particles in the system behave as if they were independent, but submitted to a common \emph{mean-field} potential due to all the other particles.

The purpose of this paper is to prove that Hartree's theory gives a correct approximation of the ground state of the many-body bosonic system in the limit of large particle number $N$ with the intensity of the pair interaction assumed to decrease proportionally to $N ^{-1}$ (which is often called the \emph{mean-field regime}). As we will recall below, there are many results of this kind in the literature, most of them dealing with particular systems. Here we propose a novel method which allows to deal with a very large class of many-body systems and does not depend on the special form of the interactions. We will be particularly interested in the case where some particles are allowed to escape to infinity, leading to a possible lack of compactness.

\subsubsection*{\bf The model}
Let us consider a system composed of $N$ identical bosons. The one-particle space is any separable Hilbert space $\gH$ and the whole system is therefore described by the $N$-fold symmetric tensor product $\gH^N:=\bigotimes_s^N\gH$. We assume that the $N$-body Hamiltonian takes the following form:
\begin{equation}
H_N:=\sum_{j=1}^N T_j+\frac{1}{N-1}\sum_{1\leq k<\ell\leq N}w_{k\ell}.
\label{eq:intro hamil} 
\end{equation}
Here $T$ is a self-adjoint operator on $\gH$ that accounts for the self energy of the particles, and $w$ is a symmetric operator on the two-particle space $\gH^2$, which corresponds to the pair interactions between the particles. As usual, $T_j = 1 \otimes \ldots \otimes T \otimes \ldots \otimes 1$ denotes the associated operator acting on the $j$-th particle and $w_{k \ell}$ is the potential acting on the pair $(k,\ell)$ of particles. 
We will always assume that $H_N$ is bounded from below.\footnote{In the mean-field regime, this automatically implies $H_N\geq -CN$.} Note that, even if we stick to Hamiltonians of the form~\eqref{eq:intro hamil} for simplicity, most of our results will indeed be valid for a much larger class of mean-field Hamiltonians. 

The fact that we are considering the mean-field regime is apparent in the factor $1/(N-1)$ in front of the interaction term in \eqref{eq:intro hamil}. It has the effect of keeping the single particle energy and the interaction energy of the same order of magnitude, so that one may expect a well-defined limit problem. Note that this factor could be replaced by any constant behaving like $1/N$ in the limit $N\to\ii$, without changing the result; the use of $1/(N-1)$ only simplifies some expressions. While this is certainly not the only scaling one may consider, it is simple and instructive, and has been very often considered in the past as a model case for the rigorous derivation of mean-field theories in many-body physics.

For physically relevant examples, we may typically think of bosons living in a bounded set $\Omega\subset \R^d$ and $T=-\Delta$ in $\gH=L^2(\Omega)$ with appropriate boundary conditions, or think of particles in $\R^d$ and $T=-\Delta +V(x)$ in $\gH=L^2(\R^d)$. In the latter case, either $V(x)\to \infty$ as $|x|\to \infty$ ($V$ is then called a trapping potential), or $V(x)\to 0$ as $|x|\to \infty$ but $V(x)$ is negative somewhere to bind (some of) the particles. The one-body operator $T$ may also involve the magnetic Laplacian $\left( -i\nabla + A(x) \right) ^2 +V(x)$ corresponding to a given vector potential $A$ on $\R ^3$, or a pseudo-relativistic operator $\sqrt{m^2c^4-\Delta} - mc^2+V(x)$. The two-body potential $w$ is often the multiplication operator by an even real-valued function $w(x-y)$ that decays at infinity. 

The Hartree functional is obtained by restricting the energy functional (quadratic form) of $H_N$ to uncorrelated functions of the form $\Psi=u^{\otimes N}$, where $u$ is a normalized vector of $\gH$. This leads to the nonlinear Hartree energy
\begin{equation}\label{eq:intro Hartree}
\frac{\pscal{u^{\otimes N},H_Nu^{\otimes N}}}N=\pscal{u,Tu}_\gH+\frac12\pscal{u\otimes u,w\, u\otimes u}_{\gH^2}:=\cEH(u).
\end{equation}
If $E(N)=\inf \sigma(H_N)$ denotes the bottom of the spectrum (ground state energy) of $H_N$, it is then expected that 
\begin{equation}\label{eq:intro hartree lim}
\lim_{N\to \infty} \frac{E(N)}{N} = \HartE 
\end{equation}
where $\HartE$ is the minimal Hartree energy 
\begin{equation}\label{eq:intro defi hartree}
\eH:=\inf_{\substack{\norm{u}=1}}\cEH(u).
\end{equation}
Since $\Psi=u^{\otimes N}$ can be used as a trial state, it is obvious that $E(N)N^{-1}\le e_{\rm H}$ for all $N\geq2$. The lower bound is much more subtle and it means that the purely uncorrelated ansatz $u^{\otimes N}$ does capture the first order of the ground state energy in the limit $N\to\ii$. We will come back below to the consequences it may have on the ground state of $H_N$ (when it exists), such as Bose-Einstein condensation.

Our aim in this paper is to provide a general strategy to justify the convergence \eqref{eq:intro hartree lim}, which is applicable to a very large class of models. The proof of \eqref{eq:intro hartree lim} is available in the literature for numerous special cases, including ``bosonic atoms'' \cite{BenLie-83,Solovej-90,Bach-91,BacLewLieSie-93,Kiessling-12}, boson stars \cite{LieThi-84,LieYau-87}, the homogeneous Bose gas \cite{Seiringer-11}, trapped Bose gases \cite{GreSei-13}, the Lieb-Liniger model \cite{LieLin-63,SeiYngZag-12}, and many others. More abstract models are discussed in \cite{FanSpoVer-80,VdBLewPul-88,RagWer-89,Werner-92}. The experimental observation of Bose-Einstein condensates in cold atomic gases has motivated a lot of interest for models of many-body bosonic systems. In this context the mean-field limit can be considered as a toy model, which is easier to analyze than the Gross-Pitaevskii limit~\cite{LieSeiSolYng-05,LieSei-06}. The latter limit corresponds to the case of a dilute gas. We will not consider it here but we hope that our method might in the future be useful to deal with it as well.

It is striking that the Hartree approximation is valid in very different physical situations. For example, in bosonic atoms the interactions are repulsive and the particles are submitted to an attractive potential generated by fixed nuclei, whereas in boson stars the interactions are attractive and the system is fully translation-invariant. This generality holds despite the fact that most proofs rely on specific properties of the Hamiltonian (and in particular of the two-body potential $w$). The main message of the present article is that the validity of the Hartree approximation does actually not rely on any specific properties of the Hamiltonian, but is rather a consequence of the \emph{special structure of the set of bosonic density matrices for large $N$}. 

Before explaining this, let us insist on the fact that we are interested here in the large-$N$ behavior of the ground state energy $E(N)$ of the Hamiltonian. There are many works on the related (but still different) derivation of the time-dependent Hartree theory from the time-dependent Schrödinger equation associated with the Hamiltonian $H_N$, see for instance~\cite{Hepp-74,GinVel-79,Spohn-80,BarGolMau-00,ElgErdSchYau-06,ElgSch-07,AmmNie-08,ErdSchYau-09,FroKnoSch-09,RodSch-09,KnoPic-10,Pickl-11}. In this case one starts close to a Hartree state at time zero, and then proves that the Schr\"odinger flow stays close to the corresponding trajectory of the Hartree state. It is fair to say that the validity of the Hartree approximation has been proved under much more general assumptions in the time-dependent case than for the ground state energy $E(N)$. Our work will therefore place the time-independent problem on the same footing as the time-dependent problem.

\subsubsection*{\bf Representability and de Finetti theorems} 
Our starting point is the formulation of the problem in terms of reduced density matrices~\cite{LieSei-09}. The $k$-particle density matrix $\gamma_\Psi^{(k)}$ of a pure $N$-body state $\Psi\in\gH^N$ is defined by 
\begin{equation}\label{eq:intro hierarchy}
\gamma^{(k)}_{\Psi}:=\tr_{k+1\to N}|\Psi\rangle\langle\Psi|
\end{equation}
for $0\leq k\leq N$, where $\tr_{k+1\to N}$ denotes the partial trace with respect to the last $N-k$ variables and $|\Psi\rangle\langle\Psi|$ is the orthogonal projection onto the state $\Psi$. Thus $\gamma^{(k)}_{\Psi}$ is a positive trace class operator on $\gH ^k$ with $\Tr_{\gH^k} \gamma^{(k)}_{\Psi}=1$ (note the normalization convention). 

The main interest of density matrices is that the energy per particle can be expressed\footnote{And of course, more generally, a problem involving at most $k$-particle interactions only depends on the $k$-particle density matrix. We stick to the two body case for clarity.}$^{,}$\footnote{Here and elsewhere in the article, the expression $\Tr \big(T\gamma_\Psi^{(1)}\big)$ should be understood in the quadratic form sense as $\Tr( \sqrt{T+C} \gamma_\Psi^{(1)} \sqrt{T+C})-C \Tr \gamma_\Psi^{(1)}$, where $C=\inf\sigma(T)$.} only in terms of $\gamma_\Psi^{(2)}$:
\begin{equation}\label{eq:intro ener matrices}
\frac{\bral \Psi, H_N \Psi \ketr}{N} = \tr_\gH \left( T \gammaP ^{(1)} \right) + \frac{1}{2} \tr_{\gH^2} \left( w \gammaP ^{(2)} \right) = \frac{1}{2} \Tr_{\gH^2} (H_2 \gamma^{(2)}_{\Psi}).
\end{equation}
We see that, thanks to the mean-field factor $1/(N-1)$ in front of the interaction,  the expression of the energy is even completely independent of $N$. The $N$ dependence is hidden in the constraint that $\gamma_\Psi^{(2)}$ must arise from an $N$-body state $\Psi$. One may thus reformulate the ground state energy of our system as 
\begin{equation}
 \label{eq:intro ener min matrices}
\frac{E(N)}{N}= \frac{1}{2}\inf\left\{ \tr_{\gH^2} \left( H_ 2 \gamma ^{(2)}\right) ,\: \gamma ^{(2)} \in \tilde\cP^{(2)} _{N}   \right\}
\end{equation}
where 
\begin{equation*}
\tilde\cP^{(2)}_N = \left\{ \gamma ^{(2)} \in \gS ^1 (\gH^2)\ :\ \exists \: \Psi \in \gH ^N, \ \norm{\Psi} = 1,\ \gamma ^{(2)} = \gammaP ^{(2)}  \right\}
\end{equation*}
is the set of all the two-particle density matrices arising from a pure $N$-body state $\Psi$ (``$N$-\emph{representable}'' two-particle density matrices). Here $\gS^1(\mathfrak{K})$ is the space of all trace-class operators on the Hilbert space $\mathfrak K$. 

It is often very useful to work with \emph{mixed states} instead of pure states. A mixed state is an operator $G$ on $\gH^N$ which is a  convex combination of pure states: $G=\sum_{i} n_i|\Psi_i\rangle\langle\Psi_i|$ with $n_i\geq0$ and $\sum_i n_i=1$. As the density matrices $\gamma_\Psi^{(k)}$ depend linearly on the operator $|\Psi\rangle\langle\Psi|$, their definition can easily be extended to mixed states:
\begin{equation}\label{eq:intro hierarchy mixed}
\gamma^{(k)}_{G}:=\tr_{k+1\to N} G.
\end{equation}
The $N$-body energy of a mixed state is now
\begin{equation}\label{eq:intro ener matrices mixed}
\frac{\tr\big(H_N G \big)}{N} = \tr_\gH \left( T \gamma_G^{(1)} \right) + \frac{1}{2} \tr_{\gH^2} \left( w \gamma_G^{(2)} \right) = \frac{1}{2} \Tr_{\gH^2} (H_2 \gamma^{(2)}_{G}).
\end{equation}
Diagonalizing $G$ and using that the energy is linear in the two-body density matrix, one sees that minimizing over mixed states gives the same answer as if one minimizes over pure states only. Therefore, we can also write
\begin{equation}
 \label{eq:intro ener min matrices mixed}
\frac{E(N)}{N}= \frac{1}{2}\inf\left\{ \tr_{\gH^2} \left( H_ 2 \gamma ^{(2)}\right) ,\: \gamma ^{(2)} \in \cP^{(2)} _{N}   \right\}
\end{equation}
where
\begin{equation*}
\cP^{(2)}_N = \left\{ \gamma ^{(2)} \in \gS^1 (\gH^2)\ :\ \exists \: 0\leq G\in\gS^1(\gH^N),\ \tr_{\gH^N}G=1,\ \gamma ^{(2)} = \gamma_G ^{(2)}  \right\}
\end{equation*}
is the set of mixed-representable two-particle density matrices, which coincides with the convex hull of the set $\tilde\cP^{(2)}_N$. Of course we may define in a similar fashion the sets $\tilde\cP^{(k)}_N$ and $\cP^{(k)}_N$ of $k$-particle density matrices, arising from pure and mixed states, respectively.

For fixed $N$, rewriting the ground state energy as in~\eqref{eq:intro ener min matrices} or~\eqref{eq:intro ener min matrices mixed} is not particularly helpful. While the one-particle set $\cP^{(1)}_N$ is known to be the set of all positive trace class operators on $\gH$ with trace $1$, it is indeed a famous open problem to characterize the set $\cP^{(2)}_N$. This is called the \emph{$N$-representability problem}, usually stated for fermions~\cite{ColYuk-00}. For bosons, it is possible to describe the sets $\cP^{(k)}_N$ in the limit $N\to\ii$, as we now explain.

Taking partial traces it is easy to see that the sets $\cP ^{(k)}_N$ form a decreasing sequence: 
$$\cP_{N+1}^{(k)}\subset \cP_N^{(k)}.$$
One may then see Problem \eqref{eq:intro ener min matrices mixed} as the minimization of a fixed energy functional on a variational set that gets more and more constrained as $N$ increases. The energy $E(N)/N$ is thus increasing with $N$ and with a leap of faith one may hope that our variational problem will converge to the one posed on the limit of the sequence of sets $\cP ^{(k)}_N$, that is on their intersection:
\begin{equation}\label{eq:intro lim N representable}
\cP^{(k)}:=\bigcap_{N\geq1}\cP^{(k)}_N,
\end{equation}
the set of $k$-body density matrices that are $N$-representable\footnote{Note that the set $\cP^{(k)}$ is empty for fermions since $\gamma_N^{(k)}\leq {N\choose k}^{-1}$ and $\Tr \gamma_N^{(k)}=1$ with our choice of normalization.} for any $N\geq k$. If we are allowed to exchange the infimum in \eqref{eq:intro ener min matrices mixed} and the limit \eqref{eq:intro lim N representable}, we then formally obtain 
\begin{equation}
\boxed{\lim_{N\to\ii}\frac{E(N)}{N}= \frac{1}{2}\inf\left\{ \tr_{\gH^2} \left( H_ 2 \gamma ^{(2)}\right) ,\: \gamma ^{(2)} \in \cP^{(2)}\right\}.}
\label{eq:formal-limit}
\end{equation}
Provided we can pass to the limit to write \eqref{eq:formal-limit}, the validity of Hartree's theory then follows from the fact that $\cP^{(2)}$ is the convex hull of the two-particle density matrices of Hartree states, $\ketl u ^{\otimes 2} \ketr\bral u ^{\otimes 2} \brar$, as we now explain.

Describing the structure of the limiting sets $\cP^{(k)}$ is precisely the object of the so-called \emph{quantum de Finetti theorem}, proved by St{\o}rmer and Hudson-Moody in~\cite{Stormer-69,HudMoo-75} and recalled in Theorem~\ref{thm:DeFinetti} below. This result is a quantum generalization of the famous classical de Finetti, also called Hewitt-Savage, theorem \cite{DeFinetti-31,DeFinetti-37,Dynkin-53,HewSav-55,DiaFre-80} about symmetric probability measures having an infinite number of variables. The importance of such results for mean-field theory has been known for a long time in the context of classical statistical mechanics~\cite{BraHep-77,Spohn-81,MesSpo-82,CagLioMarPul-92,Kiessling-93}. In~\cite{VdBLewPul-86,VdBLewPul-88,VdBDorLewPul-90} Varadhan's large deviation principle (similar in spirit to the \emph{classical} de Finetti theorem) was used to understand the Bose-Einstein condensation of certain quantum systems. The quantum version of the de Finetti theorem was then used to treat a larger class 
of problems in~\cite{FanSpoVer-80,PetRagVer-89,RagWer-89,Werner-92}.

In our language, the quantum de Finetti theorem simply states that, for any fixed $k$, \emph{the limiting set $\cP^{(k)}$ is the convex hull of the $k$-particle density matrices of Hartree states}, the latter being its extremal points. The operators $\gamma^{(k)}\in \cP^{(k)}$ can therefore all be written in the form
$$\gamma^{(k)}=\int_{S\gH}|u^{\otimes k}\rangle\langle u^{\otimes k}|\,d\mu(u)$$
where $\mu$ is a Borel probability measure on the sphere $S\gH$ of the one-particle Hilbert space $\gH$. 
Now, we can compute the right side of~\eqref{eq:formal-limit}: 
\begin{equation*}
\frac12\tr_{\gH^2} \left( H_ 2 \gamma ^{(2)}\right)=\frac12\int_{S\gH}\pscal{u^{\otimes 2},H_2 u^{\otimes 2}}\,d\mu(u)=\int_{S\gH}\cEH(u)\,d\mu(u)\geq \eH
\end{equation*}
where in the last inequality it is used that $\mu$ is a probability measure. Hence the right side of~\eqref{eq:formal-limit} is nothing but $\eH$. 

Of course, in the argument that we have sketched above, the main difficulty is to justify the formal limit~\eqref{eq:formal-limit}. A typical case that can be easily dealt with is that of trapped particles, that is, when the single-particle Hamiltonian $T$ has a compact resolvent. Simple examples consist of non-relativistic particles living in a bounded domain, or $T=-\Delta +V$ with $V$ a trapping potential. Proceeding as sketched above, one can easily justify Hartree's approximation for trapped bosons. We quickly study this situation in Section \ref{sec:trapped}. This is very much in the spirit of the earlier works \cite{FanSpoVer-80,RagWer-89} and we include it mainly for pedagogical purposes. Interesting cases covered by this approach include the homogeneous and trapped Bose gases.

\medskip

In many practical cases, however, the particles are not all trapped and some can escape to infinity. It should then be verified that those escaping to infinity are still correctly described by Hartree's theory. The situation is therefore much more complex and a more detailed analysis is necessary. It is the main object of this article to provide a strategy to carry over this detailed analysis.

In infinite dimensions, one has to be very careful of the topology which is used for investigating the limit of $\cP^{(k)}_N$. The set $\cP^{(k)}$ defined in~\eqref{eq:intro lim N representable} is actually the limit of $\cP_N^{(k)}$ for the trace norm. However, for unconfined systems it is often useful to use a weak topology instead of the strong one. The trace-class $\gS^1 (\gH^k)$ is the dual of the space $\cK(\gH^k)$ of compact operators, which is separable. Hence $\gS^1 (\gH^k)$ can as well be endowed with the corresponding weak-$\ast$ topology and its unit ball is sequentially compact for this topology. A natural question which arises in the case of a lack of compactness is then that of the \emph{weak-$\ast$ limit} of the sets $\cP^{(k)}_N$. We therefore introduce the set
\begin{equation}
\cP_{\rm w}^{(k)}:=\left\{ \gamma ^{(2)} \in \gS^1 (\gH^k),\ :\ \exists \: \gamma_{N_j}^{(k)}\in\cP^{(k)}_{N_j},\  \gamma_{N_j}^{(k)}\wto_\ast \gamma^{(k)} ~\text{as}~N_j\to \infty\right\}.
\end{equation}
In Section~\ref{sec:de_Finetti} we will prove a weak version of the quantum de Finetti theorem (see Theorem~\ref{thm:weak-De-Finetti}) which implies that $\cP_{\rm w}^{(k)}$ is the convex hull of all the weak limits of $k$-particle density matrices of Hartree states and, therefore, any $\gamma^{(k)}\in\cP_{\rm w}^{(k)}$ can be written in the form
\begin{equation}
\gamma ^{(k)} = \int_{B\gH}|u^{\otimes k}\rangle\langle u^{\otimes k}|\,d\mu(u)
\label{eq:melange_intro} 
\end{equation}
where $\mu$ is now a Borel probability measure on the \emph{unit ball} $B\gH=\{u\in\gH\ :\ \norm{u}\leq1\}$ of the one-particle Hilbert space $\gH$, instead of the unit sphere.

Our Theorem~\ref{thm:weak-De-Finetti} below is actually stronger and it states that if we have a sequence of $N$-body states such that $\gamma_{N_j}^{(k)}\wto_*\gamma^{(k)}$ for all $k\geq1$, then the limiting density matrices $\gamma^{(k)}$ can all be written as in~\eqref{eq:melange_intro} and they \emph{all share the same measure $\mu$}. By using this fact, we can easily prove the validity of Hartree's theory for systems in which the particles escaping to infinity carry a non-negative energy, that is, when the energy is a weakly lower semi-continuous function of the one- and two-particle density matrices. This is ensured if the particles may not form any bound states at infinity. We give interesting examples of such systems in Section~\ref{sec:wlsc} below. In particular, we are able to provide a short proof for bosonic atoms, as considered first by Benguria and Lieb in~\cite{BenLie-83}.

The weak version of the quantum de Finetti theorem is related to recent results of Ammari and Nier~\cite{AmmNie-08,AmmNie-11} (see Section \ref{sec:de_Finetti} for a more precise discussion), but we follow a different approach, based on geometric methods for many-body systems~\cite{Lewin-11}. Our proof of the weak quantum de Finetti theorem is in fact only the first step towards a more precise understanding of the lack of compactness for general systems and we will repeatedly use more refined arguments in the paper. 

Indeed, when the particles escaping to infinity have a nontrivial behavior, due for instance to attractive interaction potentials, looking at the weak limits of density matrices is not at all sufficient. The set $\cP_{\rm w}^{(k)}$ somehow only describe the particles which have not escaped, and the information on the other ones is completely lost. The accurate description of the lack of compactness will be done in this article, using the geometric methods of~\cite{Lewin-11}. These couple the (somehow algebraic) properties of many-particle systems with techniques from nonlinear analysis in the spirit of the concentration-compactness theory~\cite{Lieb-83,Lions-84}. Our approach is now very general and it allows to cover many quantum systems, independently of the special form of their interaction. This is the main achievement of this article. For instance, we will recover the famous result of Lieb and Yau on boson stars~\cite{LieYau-87} without using any particular property of the Newton potential. Note that 
our method is based on compactness arguments and it does not give any quantitative estimate on the discrepancy between the full many-body problem and its mean-field approximation, in contrast with operator-based methods that use specific properties of the interaction (see for example \cite{BenLie-83,Seiringer-11,LieYau-87,SeiYngZag-12}). 

\subsubsection*{\bf Typical result}
For the convenience of the reader, we state now a typical result that can be obtained from our method, and which we will prove in Section~\ref{sec:non-trapped} below. 
We consider the $N$-body Hamiltonian
\begin{equation}
H_N^V=\sum_{j=1}^N\big(\big(m^2-\Delta_{x_j})^{s}-m^{2s}+V(x_j)\big)+\frac{1}{N-1}\sum_{1\leq k<\ell\leq N}w(x_k-x_\ell)
\label{eq:H_N_intro} 
\end{equation}
on the bosonic space $\gH^N=\bigotimes_s^NL^2(\R^d)$ with $d\ge 1$. Here $m\ge 0$ and $s\in (0,1]$ are given constants. The case $s=1$ corresponds to non-relativistic particles, whereas $s=1/2$ describes a pseudo-relativistic system similar to the boson stars studied in~\cite{LieThi-84,LieYau-87}. To make $H_N^V$  a symmetric operator on $\gH^N$, as usual we assume that 
\bq \label{eq:assumption-V-w-1}
V:\R^d \to \R~~\text{\it and}~~w:\R^d \to \R~~\text{\it are measurable and}~ w(x)=w(-x).
\eq
The case of confined systems is the simplest, as explained above, and we refer to Section \ref{sec:trapped} where our results in this case are stated. Here we think of $H_N ^V$ as describing an unconfined system so both $V$ and $w$ decay at infinity, which we formalize in the following

\medskip

\noindent \textbf{Assumption on the decay at infinity of $V$ and $w$.}
{\it There exists some $R>0$ such that
\begin{equation}
V\1_{\R^d\setminus B_R}=f_1+f_2\quad\text{and}\quad
w\1_{\R^d\setminus B_R}=f_3+f_4,
\label{eq:assumption-V-w-2}
\end{equation}
where $f_j\in L^{p_j}(\R^d)$ with either $\max\{1,d/(2s)\}< p_j<\ii$, or $p_j=\ii$ and $f_j\to0$ at infinity, in the sense that $|\{ x\in \R^d: |f_j(x)|>\eps  
\}| <\infty$ for all $\eps>0$.}

\medskip

We also need to specify which possible local singularities $V$ and $w$ may have, which is the subject of the following 

\medskip

\noindent \textbf{Assumption on the local singularities of $V$ and $w$.}
{\it There exists non-negative constants $C$, $\alpha_\pm$ and $\beta_\pm$, with $\alpha_- + \beta_-<1$, such that
\begin{equation}\label{eq:assumption-V-w-3}
V_{\pm}(x)\leq \alpha_\pm (-\Delta)^{s/2}+C\quad\text{and}\quad w_{\pm}(x)\leq \beta_\pm (-\Delta)^{s/2}+C,
\end{equation}
where $f_+=\max\{f,0\}$ and $f_-=\max\{-f,0\}$ are respectively the positive and negative parts of $f$. 
}

\medskip

Since, by~\eqref{eq:assumption-V-w-2}, $V$ and $w$ are subcritical outside of the ball $B_R$, the bounds in ~\eqref{eq:assumption-V-w-3} are only interesting for the local parts $V\1_{B_R}$ and $w\1_{B_R}$. Note that under our assumption~\eqref{eq:assumption-V-w-3}, the local singularities are allowed to be comparable to the kinetic energy and, in particular, $V$ is not necessarily a compact perturbation of the kinetic operator. The upper bounds on $V_+$ and $w_+$ are not really necessary but they simplify the presentation. 

It is instructive to think of the caricature where both $V$ and $w$ are smooth functions of compact supports, which obviously satisfy our assumptions. The validity of Hartree's theory in this simple case is already a non trivial problem and does not seem to have been proven before. The conditions stated above are much more general however. In particular they are satisfied by Newton or Coulomb potentials when $s\geq1/2$. 

Under the previous assumptions on $V$ and $w$, the Hamiltonian $H^V_N$ is bounded from below and we denote by $E^V(N)$ its ground state energy. The corresponding Hartree functional reads
\begin{multline*}
\cEH^V(u)=\pscal{u,\left(\big(m^2-\Delta)^{s/2}-m^s+V\right)u}\\+\frac12\int_{\R^d}\int_{\R^d}w(x-y)|u(x)|^2|u(y)|^2\,dx\,dy 
\end{multline*}
and we denote by $\eH^V(\lambda)$ its infimum on the sphere of radius $\sqrt{\lambda}$:
$$\eH^V(\lambda)=\inf_{\substack{u\in H^s(\R^d)\\ \norm{u}_{L ^2}^2=\lambda}}\cEH^V(u).$$

Our main result contains two parts. The first item $(i)$ deals with the validity of Hartree's theory at the level of the energy, independently of the strength of the external potential $V$ (which may as well be $V\equiv0$). The second and third items of the statement give precisions about the density matrices of any sequence of approximate ground states, in particular when $V$ is sufficiently negative to bind some (or all) of the particles.

\begin{theorem}[\textbf{Validity of Hartree's theory}]\label{thm:general_intro}\mbox{}\\
Assume that $V$ and $w$ satisfy the previous assumptions~\eqref{eq:assumption-V-w-1},~\eqref{eq:assumption-V-w-2} and~\eqref{eq:assumption-V-w-3}.

\medskip

\noindent$(i)$ We always have
\begin{equation}
\boxed{\lim_{N\to\ii}\frac{E^V(N)}{N}=\eH^V(1).}
\label{eq:limit_energy_intro}
\end{equation}

\medskip

\noindent $(ii)$ Denote by $\Psi_N$ a sequence of approximate (normalized) ground states in $\gH^N$, that is, such that $\pscal{\Psi_N,H_N^V\Psi_N}= E^V(N)+o(N)$, and by $\gamma^{(k)}_N$ the corresponding density matrices. Then there exists a subsequence $(N_j)_{j\geq 1}$ and a Borel probability measure $\mu$ on the unit ball $B\gH=\{u\in\gH\; :\; \norm{u}\leq1\}$, supported on the set 
\begin{equation}
\cM^V=\Big\{u\in B\gH\ :\ \cEH^V(u)=e^V_H(\norm{u}^2)=e^V_H(1)-e^0_H(1-\norm{u}^2)\Big\},
\label{eq:def_M_V_intro} 
\end{equation}
such that 
\begin{equation}
\boxed{\gamma^{(k)}_{N_j}\wto_*\int_{\cM^V}|u^{\otimes k}\rangle\langle u^{\otimes k}|\,d\mu(u)}
\label{eq:limit_weak_intro}
\end{equation}
weakly-$\ast$ in $\gS^1(\gH^k)$, for all $k\geq1$.

\medskip

\noindent $(iii)$ Assume now that the binding inequality 
\begin{equation}
\eH^V(1)<\eH^V(\lambda)+\eH^0(1-\lambda)
\label{eq:no_dichotomy}
\end{equation}
is satisfied for all $0\leq\lambda<1$. Then the previous measure $\mu$ is supported on $S\gH$ and the limit~\eqref{eq:limit_weak_intro} for $\gamma^{(k)}_{N_j}$ is strong in the trace-class. In particular, if $\eH^V(1)$ admits a unique minimizer $u_H$, up to a phase, then there is \emph{complete Bose-Einstein condensation} on it:
\begin{equation}
\boxed{\gamma^{(k)}_{N_j}\to|u_H^{\otimes k}\rangle\langle u_H^{\otimes k}| \mbox{  strongly in } \gS ^1(\gH ^k)} 
\label{eq:limit_BEC}
\end{equation} 
for any fixed $k\geq1$.
\end{theorem}

\begin{remark}[Generalizations]
Our approach also applies for bosons in a magnetic field corresponding to replacing the fractional Laplacian $(m^2-\Delta)^{s}$ by its magnetic version $(m^2+|\nabla+iA(x)|^2)^{s}$, see Remark~\ref{rmk:magnetic}. We only need that $|A|^{2s}$ satisfies similar assumptions as $V$. We are also able to deal with particles hoping on a lattice, see Remark~\ref{rmk:lattice}.
\end{remark}

It is a classical fact used in variational methods that the non-strict inequality
$$\eH^V(1)\leq\eH^V(\lambda)+\eH^0(1-\lambda)$$
is verified for all $0\leq\lambda\leq1$. The set $\cM^V$ defined in~\eqref{eq:def_M_V_intro} contains all the minimizers of the variational problems $\eH^V(\lambda)$ for all $\lambda\in[0,1]$ which satisfy the equality 
$$\eH^V(1)=\eH^V(\lambda)+\eH^0(1-\lambda).$$
The interpretation of this condition is that a mass $1-\lambda$ can be sent to infinity without changing the lowest energy of the system.

The role of the strict binding inequality~\eqref{eq:no_dichotomy} is precisely to prevent the particles to escape at infinity, and they are very often encountered in nonlinear models. Using Lions' terminology, the assumption~\eqref{eq:no_dichotomy} is here to avoid \emph{dichotomy}, that is, to ensure that it is not favorable to split a minimizing sequence in pieces.

In the translation-invariant case ($V=0$), the many-body Hamiltonian does not have any ground state and there are sequences of approximate ground states for which the density matrices $\gamma_N^{(k)}$ all weakly tend to $0$, even after a space translation (this is called \emph{vanishing} in Lions' terminology). This is of course not in contradiction with~\eqref{eq:limit_weak_intro} since in this case $\cM^V=\cM^0$ contains $u=0$. On the other hand, there are other sequences (made of Hartree states for instance) which converge to a ground state of the Hartree functional, when it exists.

It is interesting to note that the strict binding inequality~\eqref{eq:no_dichotomy} is only assumed for the effective Hartree theory, and that it implies the expected behavior for the many-particle states. However, it is very important not to confuse~\eqref{eq:no_dichotomy} with the corresponding binding condition for the many-particle Hamiltonian $H_N^V$. By the HVZ theorem, the infimum of the essential spectrum of $H_N^V$ is
\begin{equation}
\inf\sigma_{\rm ess}\big(H^V_N\big)=\inf_{k=1,...,N}\big(E^V(N-k)+E^0(k)\big), 
\label{eq:HVZ intro}
\end{equation}
see, for example,~\cite[Thm. 12]{Lewin-11}. In particular, we have $\inf\sigma_{\rm ess}\big(H^V_N\big)\leq E^V(N-1)$ and, applying the theorem to $E^V(N-1)$, we find that 
$$\lim_{N\to\ii}\frac{E^V(N)}{N}=\lim_{N\to\ii}\frac{\inf\sigma_{\rm ess}\big(H^V_N\big)}{N}=\eH^V(1)$$
as well. Therefore, the lowest eigenvalue $E^V(N)$ (when it exists) always behaves the same to the first order as the bottom of the essential spectrum. One has to go to the next order in the large-$N$ expansion in order to distinguish the first eigenvalue from the bottom of the essential spectrum~\cite{LewNamSerSol-13}. The Hartree binding condition~\eqref{eq:no_dichotomy} only counts whether it is interesting to send a number of order $N$ of particles to infinity, whereas the HVZ criterion deals with any number of such particles, especially a number of order one.

Note that the link between the Hartree minimizer (when it exists) and the many-particle ground states is only expressed here in terms of the density matrices $\gamma_N^{(k)}$. It is \emph{wrong} in general that an approximate ground state $\Psi_N\in\gH^N$ is close to a state $u^{\otimes N}$ for the norm of $\gH^N$. One has to go to the next order in $N$ to understand precisely the link between the two wave functions, see~\cite{LewNamSerSol-13}.

If there is complete Bose-Einstein condensation (BEC) as in~\eqref{eq:limit_BEC}, and if the unique Hartree ground state is non-degenerate, then it is shown by Lewin, Nam, Serfaty and Solovej in~\cite{LewNamSerSol-13} that the energy can be expanded as
$$E^V(N)=N\eH^V(1)+e_{\rm B}^V+o(1)$$
where $e_{\rm B}^V$ is the ground state energy of an effective operator in Fock space called the \emph{Bogoliubov Hamiltonian}. Since the validity of Hartree's theory and complete BEC were \emph{assumptions} in~\cite{LewNamSerSol-13}, the present work supplements the article~\cite{LewNamSerSol-13}.

\medskip

Before going to the more technical parts of the paper we summarize informally its main message:
\begin{enumerate}
\item The validity of Hartree's theory in the mean-field limit may be viewed as a consequence of the structure of the set of bosonic $N$-particle density matrices for large $N$. In this limit the representability problem can be given a satisfactory answer via the quantum de Finetti theorem. For confined systems, no other ingredient is needed. 
\item For unconfined systems where particles are allowed to escape to infinity, a deeper analysis is required and it can be realized by combining a weak version of the quantum de Finetti theorem together with the geometric methods for many-body systems described in~\cite{Lewin-11}.
\item Our method is general and can be applied in many different situations (non-relativistic or relativistic particles living in a domain, in the whole space or on a lattice, with or without external fields, etc).
\end{enumerate}

\medskip

\paragraph*{\bf Organization of the paper.} In the next section we discuss two versions of the quantum de Finetti theorem. After having recalled the usual statement of St{\o}rmer and Hudson-Moody, we prove in Theorem~\ref{thm:weak-De-Finetti} a weak version which will be very useful throughout the paper. In Section~\ref{sec:trapped} we quickly explain how to deal with confined systems (in which there is no lack of compactness at infinity), using the usual quantum de Finetti theorem. This is mainly introduced here for pedagogical purposes. Then, in Section~\ref{sec:non-trapped} we study the Hamiltonian~\eqref{eq:H_N_intro}. We start with the case of repulsive systems (for instance $w\geq0$), for which the energy is a weakly lower semi-continuous function of the one- and two-particle density matrices. In this case the proof is an immediate consequence of the weak quantum de Finetti theorem. This covers bosonic atoms for example. Then, in Section~\ref{sec:tr-in}, we investigate purely translation-invariant 
systems, by using some ideas of Lieb and Yau~\cite{LieYau-87}, coupled to the geometric methods of~\cite{Lewin-11}. Finally,  we prove Theorem~\ref{thm:general_intro} in Section~\ref{sec:general-proof} and discuss some of its generalizations in Section~\ref{sec:extensions}. An alternative proof of the weak de Finetti theorem is presented in Appendix~\ref{app:Hudson-Moody-proof}.

\bigskip

\noindent\textbf{Acknowledgment.} The authors acknowledge financial support from the European Research Council under the European Community's Seventh Framework Programme (FP7/2007-2013 Grant Agreement MNIQS 258023).

\section{A weak quantum de Finetti theorem and geometric localization}\label{sec:de_Finetti}

The quantum de Finetti theorem is about the structure of bosonic states on the infinite tensor product of a given algebra. A simple formulation can be given in terms of density matrices only, following~\cite{HudMoo-75}.

\begin{theorem}[\bf Quantum de Finetti]\label{thm:DeFinetti}\mbox{}\\
Let $\gH$ be any separable Hilbert space and denote by $\gH^k:=\bigotimes_s^k\gH$ the corresponding bosonic $k$-particle space. Consider a hierarchy $\{\gamma^{(k)}\}_{k=0}^\infty$ of non-negative self-adjoint operators, where each $\gamma^{(k)}$ acts on $\gH^k$. We assume that the hierarchy is \emph{consistent} in the sense that
\begin{equation}
\tr_{k+1\to k+n}\gamma^{(k+n)}=\gamma^{(k)}
\label{eq:consistent}
\end{equation}
for all $k,n\geq0$. We also assume that $\gamma^{(0)}=1$, which then implies $\tr_{\gH^k}\gamma^{(k)}=1$ for all $k\geq0$.

Then there exists a unique Borel probability measure $\mu$ on the sphere $S\gH$ of $\gH$, invariant under the group action of $S^1$, such that
\begin{equation}
\boxed{\gamma^{(k)}=\int_{S\gH}|u^{\otimes k}\rangle\langle u^{\otimes k}| \, d\mu(u)}
\label{eq:melange}
\end{equation}
for all $k\geq0$.

If moreover we have $\tr(T\gamma^{(1)})<\ii$ for some self-adjoint operator $T\geq0$ on $\gH$, then $\mu$ is supported in the quadratic form domain $Q(T)$ of $T$. That is, $\mu \big( S\gH\setminus Q(T) \big)=0$.
\end{theorem}

The result is the quantum equivalent of the famous Hewitt-Savage theorem for classical systems~\cite{DeFinetti-31,DeFinetti-37,Dynkin-53,HewSav-55,DiaFre-80,Lions-CdF}. The latter deals with a hierarchy of symmetric probability measures $\mu^{(k)}$ on $\Omega^N$ such that $\mu^{(k)}(A)=\mu^{(k+n)}(A\times\Omega^{n})$ for any $k,n\geq0$ and any measurable set $A\subset\Omega^k$. The quantum de Finetti Theorem~\ref{thm:DeFinetti} was proved in~\cite{Stormer-69,HudMoo-75} (see~\cite{Gottlieb-05,ChrKonMitRen-07} for related content). The usual statement does not include the part concerning the operator $T$, but this part is easily shown by using the Hilbert space structure associated with the quadratic form of $T$ instead of the original one. 

If we are given a sequence of states $\Psi_N\in\gH^N$ with $N\to\ii$, we obtain a sequence of density matrices, $(\gamma_{\Psi_N}^{(k)})_{0\leq k\le N}$, which form a consistent hierarchy in the sense of~\eqref{eq:consistent}, but only for $k\leq N$:
\begin{equation}\label{eq:hierarchy finite}
\begin{cases}
\gamma^{(N)}_{\Psi_N}:=|\Psi_N\rangle\langle\Psi_N|,\\
\gamma^{(k)}_{\Psi_N}:=\tr_{k+1\to N}\gamma^{(N)}_{\Psi_N}&\text{for $0\leq k\leq N$,}\\
 \gamma^{(k)}_{\Psi_N}:=0&\text{for $k\geq N+1$}.
\end{cases} 
\end{equation}
As each $\gamma_{\Psi_N}^{(k)}$ is non-negative and has a trace normalized to 1, it is bounded in the trace-class $\gS^1(\gH^k)$ for every fixed $k\ge 1$. Therefore it has a subsequence which converges weakly-$\ast$ in the trace class, $\gamma_{\Psi_{N_j}}^{(k)}\wto_*\gamma^{(k)}$. By the diagonal procedure, we can make all the density matrices weakly-$\ast$ converge along the same subsequence $N_j\to\ii$. So we obtain in the limit an infinite hierarchy of density matrices $(\gamma^{(k)})_{k\geq0}$, where each $\gamma^{(k)}$ acts on $\gH^k$.

In general the sequence $(\gamma^{(k)})_{k\geq0}$ is \emph{not} consistent because the trace is not continuous for the weak-$\ast$ topology in infinite dimension. Indeed, when passing to the weak limit we find by Fatou's lemma for trace-class operators
\begin{align*}
\gamma^{(k)}=\underset{j\to\ii}{\text{w-lim}}\;\gamma_{ \Psi_{N_j}}^{(k)}&=\underset{j\to\ii}{\text{w-lim}}\;\tr_{k+1\to k+n}\gamma_{ \Psi_{N_j}}^{(k+n)}\\&\geq\tr_{k+1\to k+n}\underset{j\to\ii}{\text{w-lim}}\;\gamma_{ \Psi_{N_j}}^{(k+n)}=\tr_{k+1\to k+n}\gamma^{(k+n)},
\end{align*}
where w-lim denotes the weak-$\ast$ limit in $\gS^1(\gH^k)$. Equality will not hold in general in the above equation and so we only have \emph{a priori}
\begin{equation}
\gamma^{(k)}\geq \tr_{k+1\to k+n}\gamma^{(k+n)},\qquad \forall n,k\geq0.
\label{eq:sub-consistent} 
\end{equation}
If $\tr(\gamma^{(1)})=1$ then it can be proved that $\tr(\gamma^{(k)})=1$ for all $k\geq1$ (see Corollary~\ref{cor:strong_CV} below). The convergence must then be strong by the reciprocal of Fatou's lemma (see for example~\cite[Add. H]{Simon-79}), and the sequence $(\gamma^{(k)})_{k\geq0}$ is consistent. In this case the quantum de Finetti Theorem~\ref{thm:DeFinetti} applies to the limiting hierarchy $(\gamma^{(k)})_{k\geq0}$.

In the systems in which particles are allowed to escape to infinity, a non-consistent hierarchy can be obtained in the limit. Consider, for instance, a given orthonormal basis $(u_j)$ of the one-particle space $\gH$,  and the Hartree state
$$\Psi_N=\big(\cos(\theta)\, u_1+\sin(\theta)\, u_N\big)^{\otimes N}$$
for some $\theta\in[0,2\pi]$. In this case we get
\begin{align*}
\gamma_{\Psi_N}^{(k)}&=\Big|\big(\cos(\theta)\, u_1+\sin(\theta)\, u_N\big)^{\otimes k}\Big\rangle\Big\langle\big(\cos(\theta)\, u_1+\sin(\theta)\, u_N\big)^{\otimes k}\Big|\\
&\underset{\ast}{\wto}\Big|\big(\cos(\theta)\, u_1\big)^{\otimes k}\Big\rangle\Big\langle\big(\cos(\theta)\, u_1\big)^{\otimes k}\Big|:=\gamma^{(k)}. 
\end{align*}
We see that in this example the sequence $(\gamma^{(k)})$ can be represented by a formula similar to~\eqref{eq:melange}, but with a measure $\mu$ that is the uniform delta measure on the circle $\varphi\in[0,2\pi)\mapsto e^{i\varphi}\cos(\theta)\,u_1$ in the unit ball 
$B\gH=\{u\in\gH\ :\ \norm{u}\leq 1\}$ of $\gH$. That we end up with a measure living on the unit ball $B\gH$ instead of the unit sphere $S\gH$ is not a big surprise, of course, as we are considering weak limits. This is actually the general case, as stated in the following result, which will be an important tool in this article.

\begin{theorem}[\bf Weak quantum de Finetti]\label{thm:weak-De-Finetti}\mbox{}\\
Let $\gH$ be any separable Hilbert space and denote by $\gH^k:=\bigotimes_s^k\gH$ the corresponding bosonic $k$-particle space. Let $\Gamma_N$ be any sequence of mixed states on $\gH^N$ (that is, $\Gamma_N\geq0$ and $\tr_{\gH^N}\Gamma_N=1$) such that 
$$\gamma^{(k)}_{\Gamma_N}\wto_\ast\gamma^{(k)}$$
weakly-$\ast$ in the trace class $\gS^1(\gH^k)$ for all $k\geq1$. Then there exists a unique Borel probability measure $\mu$ on the \emph{unit ball} $B\gH$ of $\gH$, invariant under the group action of $S^1$, such that
\begin{equation}
\boxed{\gamma^{(k)}=\int_{B\gH}d\mu(u)\;|u^{\otimes k}\rangle\langle u^{\otimes k}|}
\label{eq:melange_weak}
\end{equation}
for all $k\geq0$.

Moreover, if we have $\tr(T\gamma^{(1)})<\ii$ for some self-adjoint operator $T\geq0$ on $\gH$, then $\mu$ is supported in the quadratic form domain $Q(T)$ of $T$. That is, $\mu\big(B\gH\setminus Q(T) \big)=0$.
\end{theorem}

Ammari and Nier have recently proved in \cite{AmmNie-08,AmmNie-11} results that imply Theorem~\ref{thm:weak-De-Finetti}. In analogy with semi-classical analysis, they called $\mu$ a \emph{Wigner measure}. They deal with an arbitrary sequence of states in Fock space and, therefore, obtain in the limit a measure $\mu$ which can live over the whole one-particle Hilbert space $\gH$, instead of the unit ball as in our situation. In~\cite[Thm. 6.2]{AmmNie-08} they first construct the Wigner measure $\mu$ by testing against anti-Wick and Weyl observables, before looking at Wick observables which are related to density matrices in~\cite[Cor. 6.14]{AmmNie-08}. The case of $\mu$ having its support in a ball is studied in~\cite{AmmNie-11}. The connection between Wigner and de Finetti measures is discussed in \cite[Section 6.3]{Ammari-HDR}.

In the present paper we provide two different proofs of Theorem~\ref{thm:weak-De-Finetti}, which are both based on Theorem~\ref{thm:DeFinetti}. The first proof is based on the \emph{finite-dimensional} de Finetti Theorem and on the geometric techniques introduced in~\cite{Lewin-11} and it has the merit of clarifying how the measure $\mu$ arises in case the density matrices $\gamma_{\Psi_N}^{(k)}$ do not converge strongly. This will be particularly important to understand unconfined quantum systems in the rest of the paper and hence we explain this first approach in details in this section. It is also possible to prove Theorem~\ref{thm:weak-De-Finetti} by following arguments similar to those of Hudson and Moody~\cite{HudMoo-75}, and we quickly explain this in Appendix~\ref{app:Hudson-Moody-proof} for completeness.

\begin{remark}
It could seem an interesting question to find the structure of the set of all the density matrices satisfying the inequality~\eqref{eq:sub-consistent}. In Theorem~\ref{thm:weak-De-Finetti} we only characterize those arising from a sequence of states $(\Gamma_N)$ with $N\to\ii$. The set of density matrices satisfying~\eqref{eq:sub-consistent} is way too large, however. For instance it contains the density matrices of $n$-particle states with $n$ fixed, for which a representation of the form~\eqref{eq:melange_weak} cannot hold. Consider
$$\gamma^{(0)}=1,\qquad\gamma^{(1)}=|u\rangle\langle u|,\qquad \gamma^{(n)}\equiv0\quad \text{for $n\geq2$},$$
which of course satisfies~\eqref{eq:sub-consistent}. These density matrices are those of a one-particle state $u\in\gH$ and they cannot be written in the form~\eqref{eq:melange_weak}. The measure $\mu$ would need to be $\mu=\delta_{u}$ but then $\gamma^{(n)}\neq0$ for $n\geq2$. 
\end{remark}

By using the weak de Finetti Theorem~\ref{thm:weak-De-Finetti} we easily recover the well-known fact that the density matrices $\gamma^{(k)}_{\Gamma_N}$ all converge strongly if and only if the one-particle density matrix $\gamma^{(1)}_{\Gamma_N}$ converges strongly.

\begin{corollary}[\textbf{Strong convergence}]\label{cor:strong_CV} \mbox{}\\
Let $\{\Gamma_N\}$ be as in Theorem~\ref{thm:weak-De-Finetti}. Then the following are equivalent:
\begin{enumerate}
\item $\tr_{\gH}\gamma^{(1)}=1$;
\item $\gamma^{(1)}_{\Gamma_N}\to\gamma^{(1)}$ strongly in the trace-class $\gS^1(\gH)$;
\item $\tr_{\gH^k}\gamma^{(k)}=1$ for all $k\geq1$;
\item $\gamma^{(k)}_{\Gamma_N}\to\gamma^{(k)}$ strongly in the trace-class $\gS^1(\gH^k)$ for all $k\geq1$;
\item the measure $\mu$ of Theorem~\ref{thm:weak-De-Finetti} has its support on the sphere $S\gH$.
\end{enumerate}
\end{corollary}

\begin{proof}
The equivalence between (1) and (2) (and (3) and (4)) is the reciprocal of Fatou's lemma mentioned above (see, for instance,~\cite{dellAntonio-67},~\cite[Cor. 1]{Robinson-70} and~\cite[Add. H]{Simon-79}). The rest follows directly from~\eqref{eq:melange_weak}.
\end{proof}

\subsubsection*{\bf Geometric localization}
In this section we provide a proof of Theorem~\ref{thm:weak-De-Finetti} using the concept of \emph{geometric localization}\footnote{The term ``geometric'' was first used in~\cite{Simon-77,Sigal-82} to denote the use of partitions of unity in the configuration space for many-body systems.} in Fock space, which was developed in~\cite{DerGer-99,Ammari-04} and thoroughly used for nonlinear many-body systems in~\cite{Lewin-11}. This method is the first step to understand the lack of compactness of the sequence $(\Psi_{N})$ of approximate minimizers appearing in Theorem~\ref{thm:general_intro}.

For any mixed state $\Gamma_N$ in the $N$-particle space $\gH^N$ and any self-adjoint operator $0\leq A\leq 1$ on $\gH$, called the localizing operator, the corresponding localized state is by definition the unique state in Fock space
$$\cF(\gH)=\C\oplus\gH\oplus\gH^2\oplus\cdots$$
which has $A^{\otimes n}\gamma_{\Gamma_N}^{(n)}A^{\otimes n}$ as $n$-particle density matrices, for all $1\leq n\leq N$. The localized state lives in the truncated Fock space $\C\oplus \gH\oplus\cdots\oplus\gH^N\oplus0\cdots$ and it is often a mixed state, even when the initial state $\Gamma_N$ is pure. This localized state can be explicitly computed (see~\cite[Example 10]{Lewin-11}) and it is equal to
\begin{equation}
G_N=G_{N,0}^A\oplus G_{N,1}^A\oplus\cdots\oplus G_{N,N}^A\oplus0\oplus\cdots 
\label{eq:def_localization}
\end{equation}
where
\begin{equation}
G^A_{N,k}={N\choose k}\tr_{k+1\to N}\Big(A^{\otimes k}\otimes \sqrt{1-A^2}^{\otimes N-k}\;\Gamma_N\;A^{\otimes k}\otimes \sqrt{1-A^2}^{\otimes N-k}\Big).
\label{eq:def_localization2} 
\end{equation}
An important property of the localization of $N$-particle states is that 
\begin{equation}
\tr_{\gH^k} G^A_{N,k}=\tr_{\gH^{N-k}} G^{\sqrt{1-A^2}}_{N,N-k}
\label{eq:relation-geometric}
\end{equation}
for all $k=0,...,N$. If $A= \1_D$ for some $D\subset \R ^d$, then \eqref{eq:relation-geometric} simply means that the probability of having $k$ particles inside $D$ is equal to the probability of having $N-k$ particles outside $D$.

A simple calculation also shows that
\begin{equation}
A^{\otimes n}\gamma^{(n)}_{\Gamma_N}A^{\otimes n}={N\choose n}^{-1}\sum_{k=n}^N{k\choose n}\tr_{n+1\to k}G^A_{N,k},
\label{eq:localized-DM} 
\end{equation}
see~\cite{Lewin-11} for details. We emphasize that in~\cite{Lewin-11} a different convention is used for the $n$-particle density matrix, and this is responsible for the additional factor ${N\choose n}^{-1}{k\choose n}$. 

\begin{remark}
When $A=P$ is an orthogonal projection, geometric localization is the same as using the isometry of Fock spaces $\cF(\gH)=\cF(P\gH)\otimes\cF(P^\perp\gH)$ and restricting the state to the smaller Fock space $\cF(P\gH)$, by taking a partial trace in the second variable.
\end{remark}

The link between the de Finetti measure $\mu$ and geometric localization is emphasized in the following result, which will be very useful in the proof of Theorem~\ref{thm:general_intro} below.

\begin{theorem}[\textbf{De Finetti measure and geometric localization}]\label{thm:other-localization}\mbox{}\\
Let $(\Gamma_N)$ be as in Theorem~\ref{thm:weak-De-Finetti} and assume, furthermore, that $\tr T\gamma^{(1)}_{\Gamma_N}\leq C$ for some non-negative self-adjoint operator $T$. Let $0\leq A\leq 1$ be an operator on $\gH$ such that $A(1+T)^{-1/2}$ is compact and let $G_N^A$ be the associated localized state in the Fock space $\cF(\gH)$ defined in~\eqref{eq:def_localization} and~\eqref{eq:def_localization2}. Then
\begin{equation}
\lim_{N\to\ii}\sum_{k=0}^Nf\left(\frac{k}{N}\right)\tr_{\gH^k}G_{N,k}^A=\int_{B\gH}d\mu(u)\;f(\|Au\|^2)
\label{eq:localized_de_finetti}
\end{equation}
for all continuous functions $f$ on $[0,1]$.
\end{theorem}

Here geometric localization is used to detect the particles which do not escape to infinity and the role of the operator $A$ is to turn weak convergence into strong convergence~\cite{Lewin-11}. The interpretation of the convergence~\eqref{eq:localized_de_finetti} is that the mass of the de Finetti measure $\mu$ on the sphere $\{\|u\|^2=\lambda\}$ is the probability that a fraction $\lambda$ of the particles does not escape to infinity. 

Of course, we can turn the matter around and obtain an information on the number of particles that have escaped. Using the fundamental relation \eqref{eq:relation-geometric}, Theorem \ref{thm:other-localization} implies\footnote{The convergence~\eqref{eq:localized_de_finetti} does not hold in general if $A$ is replaced by $(1-A^2)^{1/2}$, because $(1-A^2)^{1/2}(1+T)^{-1/2}$ is not always compact. Indeed, $1- \|A u \|^2> \|\sqrt{1-A^2}u\|^2$ when $\norm{u}<1$.}
\begin{align}
\lim_{N\to\ii}\sum_{k=0}^N f\left(\frac{k}{N}\right)\tr_{\gH^k}G_{N,k}^{\sqrt{1-A^2}}&=\lim_{N\to\ii}\sum_{k=0}^N f\left(1-\frac{k}{N}\right)\tr_{\gH^k}G_{N,k}^{A}\nn\\
&=\int_{B\gH}d\mu(u)\;f(1- \|A u \|^2). \label{eq:localized_de_finetti2}
\end{align}
In practice, we will use~\eqref{eq:localized_de_finetti2} with $A=\chi_R(x)$, a localization function in a ball of radius $R$, which is compact relatively to the fractional Laplacian.

The rest of this section will be devoted to the proof of Theorem~\ref{thm:weak-De-Finetti} and Theorem~\ref{thm:other-localization} using geometric localization.

\begin{proof}[Proof of Theorem~\ref{thm:weak-De-Finetti}]
Our approach can be summarized as follows. First we localize the quantum state in a finite-dimensional space using a finite dimensional projection $P$, in order to convert the weak convergence into the strong convergence. By doing so we get the localized state $G_N^P$ in Fock space and we apply the quantum de Finetti to each of the projections $G_{N,k}^P$ in the $k$-particle spaces. The limiting density matrices have a representation in ``spherical coordinates'' on $B\gH$,
\begin{equation}
\gamma_{\rm loc}^{(n)}=\int_{0}^1\int_{S\gH}d\nu_{\rm loc}(\lambda,u)\;\lambda^n|u^{\otimes n}\rangle\langle u^{\otimes n}|,
\label{eq:melange_weak2}
\end{equation}
where $\lambda$ corresponds to all the possible values of $k/N$. In the end of the proof we remove the localization and get the result.

To carry out this program, we fix an orthogonal projection of finite rank $P$ on $\gH$ and let $G_N^P$ be the corresponding localized state as defined in~\eqref{eq:def_localization} and~\eqref{eq:def_localization2}. We write the sum in~\eqref{eq:localized-DM} as an integral over an additional parameter $0\leq\lambda\leq1$, which is $\lambda= k/N$. Let $M_{P,N}^{(n)}$ be the following Radon measure on $[0,1]$, with values in the set of self-adjoint operators in $\gS^{1}(\otimes_s^n P\gH)$ (that is, hermitian matrices of size $\dim\,\otimes_s^n (P\gH)$):
$$\dM_{P,N}^{(n)}(\lambda):=\sum_{k=n}^N\;\delta_{k/N}(\lambda)\;\tr_{n+1\to k}G^P_{N,k}.$$
This measure satisfies  
$$\int_0^1\tr_{\gH^n}\dM_{P,N}^{(n)}(\lambda)=\sum_{k=n}^N\;\tr_{\gH^k}G^P_{N,k}\leq\sum_{k=0}^N\;\tr_{\gH^k}G^P_{N,k}=\int_0^1 \dM_{P,N}^{(0)}(\lambda)=1.$$
Then, using \eqref{eq:localized-DM}, we have for all $n\geq0$
\begin{equation*}
\tr\left|P^{\otimes n}\gamma^{(n)}_{\Gamma_N}P^{\otimes n}-\int_0^1 \lambda^n \dM_{P,N}^{(n)}(\lambda)\right|
\leq  \sum_{k=n}^N\left| {N\choose n}^{-1} {k\choose n}-\left(\frac{k}{N}\right)^n\right|\tr G^P_{N,k}.
\end{equation*}
We can write
\begin{equation*}
\left(\frac{k}{N}\right)^n- {N\choose n}^{-1} {k\choose n}=\left(\frac{k}{N}\right)^n \left\{1-\frac{1-\frac{1}{k}}{1-\frac{1}{N}}\cdot\frac{1-\frac{2}{k}}{1-\frac{2}{N}}\cdots \frac{1-\frac{n-1}{k}}{1-\frac{n-1}{N}} \right\}
\end{equation*}
and by using Bernoulli's inequality 
\begin{multline*}
\prod_{j=1}^{n-1}\frac{1-\frac{j}{k}}{1-\frac{j}{N}}=\prod_{j=1}^{n-1}\left(1-\frac{j}{N-j}\left(\frac{N}k-1\right)\right)\\
\geq \left(1-\frac{n-1}{N-n+1}\left(\frac{N}k-1\right)\right)^{n-1}\geq 1-\frac{ (n-1)^2}{N-n+1}\left(\frac{N}k-1\right),
\end{multline*}
we obtain
\begin{equation}
\left(\frac{k}{N}\right)^n- {N\choose n}^{-1} {k\choose n}\leq \left(\frac{k}{N}\right)^{n}\left( \frac{N}k-1 \right)  \frac{(n-1)^2}{N-n+1}\leq \frac{(n-1)^2}{N-n+1}.
\label{eq:difference_k_N}
\end{equation}
Therefore, 
\begin{equation}
\tr\left|P^{\otimes n}\gamma^{(n)}_{\Gamma_N}P^{\otimes n}-\int_0^1 \lambda^n\,\dM_{P,N}^{(n)}(\lambda)\right| \leq\frac{(n-1)^2}{N-n+1}\sum_{k=n}^N\tr G^P_{N,k}
\leq\frac{(n-1)^2}{N-n+1}.
\end{equation}
Since $P$ is a projection of finite rank and by assumption $\gamma^{(n)}_{\Gamma_N}\wto \gamma^{(n)}$ weakly-$\ast$ in the trace-class, the sequence $P^{\otimes n}\gamma^{(n)}_{\Gamma_N}P^{\otimes n}$ converges strongly in the trace-class to $P^{\otimes n}\gamma^{(n)}P^{\otimes n}$. Therefore we have proved that, for any fixed $n$,
$$\lim_{N\to\ii}\tr\left|P^{\otimes n}\gamma^{(n)}P^{\otimes n}-\int_0^1 \lambda^n\, \dM_{P,N}^{(n)}(\lambda)\right|=0.$$
By construction, $M_{P,N}^{(n)}$ is a bounded sequence of (Radon) measures on $[0,1]$ with values in the cone of non-negative operators in the finite-dimensional space $\otimes_s^n P \gH \subset\gH^n$ having a trace $\leq1$. This sequence has a subsequence which converges weakly (in the sense of bounded measures on a compact set) to some measure $M_P^{(n)}$ and is therefore such that 
$$P^{\otimes n}\gamma^{(n)}P^{\otimes n}=\int_0^1 \lambda^n\,\dM_P^{(n)}(\lambda)$$
for all $n\geq1$. For $n=0$, $M_P^{(0)}$ is a Borel probability measure on $[0,1]$. For $n\geq1$, the value of the measure $M_P^{(n)}$ at $\lambda=0$ is not important, as it does not contribute to the density matrices. In order to simplify our reasoning, we will simply take $M_P^{(n)}(\{0\})=M_P^{(0)}(\{0\})\,{\rm Id}_{\otimes_s^n P\gH}\tr\big({\rm Id}_{\otimes_s^n P \gH}\big)^{-1}$.

Now, we claim that the sequence $M_P^{(n)}$ is consistent in the sense that 
\begin{equation}
\tr_{n+1\to n+k}M_P^{(n+k)}= M_P^{(n)}
\label{eq:A_P_consistent}
\end{equation}
for $n,k\geq0$. Indeed, we have
\bqq
\tr_{n+1\to n+k}M_{P,N}^{(n+k)}(\lambda)&=&\sum_{j=n+k}^N\;\delta_{j/N}(\lambda)\;\tr_{n+1\to j}G_{N,j}\hfill\\
&=&M_{P,N}^{(n)}(\lambda)-\sum_{j=n}^{n+k-1}\;\delta_{j/N}(\lambda)\;\tr_{n+1\to j}G^P_{N,j}
\eqq
and therefore
$$\lambda\tr_{\gH^n}\left|\tr_{n+1\to n+k}M_{P,N}^{(n+k)}(\lambda)-M_{P,N}^{(n)}(\lambda)\right|\leq \frac{n+k-1}{N}\sum_{j=n}^{n+k-1}\;\delta_{j/N}(\lambda)\;\tr G^P_{N,j}.$$
The sum of the right side is a uniformly bounded measure on $[0,1]$, hence we obtain~\eqref{eq:A_P_consistent}.

By the (finite-dimensional) quantum de Finetti Theorem~\ref{thm:DeFinetti}, there exists a Borel probability measure $\nu_P$ on $[0,1]\times S\gH\cap (P\gH)$, invariant under the action of $S^1$, such that 
$$\dM_P^{(n)}(\lambda)=\int_{S\gH}d\nu_P(\lambda,u)|u^{\otimes n}\rangle\langle u^{\otimes n}|.$$
The original statement applies only for each fixed $\lambda\in[0,1]$. In order to deal with the present case of a probability measure on $[0,1]$, we first approximate $M^{(n)}_P(\lambda)$ by a step function and then we pass to the limit. The weak limit is still a Borel probability measure as we are working on a compact set of a finite dimensional space. We have now proved that 
\bqq
P^{\otimes n}\gamma^{(n)}P^{\otimes n}&=&\int_{0}^1\int_{S\gH}d\nu_{P}(\lambda,u)\,\lambda^n|u^{\otimes n}\rangle\langle u^{\otimes n}|\hfill\\
&=&\int_{0}^1\int_{S\gH}d\nu_{P}(\lambda,u)\,|(\sqrt\lambda u)^{\otimes n}\rangle\langle (\sqrt\lambda u)^{\otimes n}|.
\eqq
Associated to the probability measure $\nu_P$ on $[0,1]\times S\gH\cap P\gH$, there is a unique probability measure $\mu_P$ on $B\gH\cap P\gH$ which is formally given by the formula $\mu_P(u):=\nu_P(\norm{u},u\norm{u}^{-1})$. More precisely, $\mu_P$ is defined by its action on continuous functions on $B\gH$ by
$$\int_{B\gH}f(u)\,d\mu_P(u):=\int_0^1\int_{S\gH}f(\sqrt{\lambda} u)\,d\nu_P(\lambda,u)$$
and we get as we wanted 
$$P^{\otimes n}\gamma^{(n)}P^{\otimes n}=\int_{B\gH}d\mu_P(u)\;|u^{\otimes n}\rangle\langle u^{\otimes n}|$$
for all $n\geq1$.

The argument can be applied for any chosen finite-rank orthogonal projection $P$. Consider now a sequence $P_k$ which converges strongly to the identity on $\gH$ and apply the previous argument for each $k$. We find a sequence of Borel probability measures $\mu_k$ on $BP_k\gH$, invariant under the action of $S^1$, such that 
$$P_k^{\otimes n}\gamma^{(n)}P_k^{\otimes n}=\int_{B\gH}d\mu_k(u)\,|u^{\otimes n}\rangle\langle u^{\otimes n}|.$$
By construction the measure $\mu_k$ coincides with $\mu_\ell$ for $\ell\leq k$ on cylindrical Borel sets having their base in $P_\ell\gH$. Since these measures all have their support in a bounded set in $\gH$, there exists by~\cite[Lemma 1]{Skorokhod-74}  a unique Borel probability measure $\mu$ on $B\gH$ which coincides with $\mu_k$ on cylindrical Borel subsets having their base in $P_k\gH$. This precisely means that 
$$\int_{B\gH}d\mu_k(u)\,|u^{\otimes n}\rangle\langle u^{\otimes n}|=\int_{B\gH}d\mu(u)\,|(P_ku)^{\otimes n}\rangle\langle (P_ku)^{\otimes n}|$$
and therefore
$$P_k^{\otimes n}\gamma^{(n)}P_k^{\otimes n}=P_k^{\otimes n}\left(\int_{B\gH}d\mu(u)\,|u^{\otimes n}\rangle\langle u^{\otimes n}|\right)P_k^{\otimes n}.$$
Taking now $k\to\ii$ finishes the proof of the existence of $\mu$.

We conclude with the uniqueness of $\mu$ and assume that another $S^1$--invariant Borel probability measure $\mu'$ on $B\gH$ satisfies
\begin{equation}
\int_{B\gH}|u^{\otimes k}\rangle\langle u^{\otimes k}|d\mu(u)=\int_{B\gH}|u^{\otimes k}\rangle\langle u^{\otimes k}|d\mu'(u) 
\label{eq:uniqueness}
\end{equation}
for all $k\geq1$. From~\cite[Lemma 1]{Skorokhod-74}, $\mu$ is characterized by its cylindrical projections $\mu_V$ on any finite-dimensional subspace $V\subset\gH$ and therefore it suffices to prove that $\mu_V=\mu_V'$ for all such $V$'s. Taking the projection $(P_V)^{\otimes k}$ on both sides of~\eqref{eq:uniqueness}, we see that $\mu_V$ and $\mu'_V$ satisfy the same relation as in~\eqref{eq:uniqueness}, with $B\gH$ replaced by the unit ball $BV$ of $V$. Let then $(e_1,...,e_d)$ be an orthonormal basis of $V$ and $P_i=|e_i\rangle\langle e_i|$ be the associated projections. Applying $P_{i_1}\otimes\cdots\otimes P_{i_k}$ on the left and $P_{j_1}\otimes\cdots\otimes P_{j_k}$ on the right side of~\eqref{eq:uniqueness}, we get 
$$\int_{BV}u_{i_1}\cdots u_{i_k}\overline{u_{j_1}}\cdots \overline{u_{j_k}}\,d(\mu_V-\mu_V')(u)=0$$
for all multi-indices $i_1,...,i_k$ and $j_1,...,j_k$ (with $u=\sum_{j=1}^du_je_j$). From the $S^1$--invariance of the two measures, it is clear that, for $k\neq \ell$,
$$\int_{BV}u_{i_1}\cdots u_{i_k}\overline{u_{j_1}}\cdots \overline{u_{j_\ell}}\,d(\mu_V-\mu_V')(u)=0.$$
Since polynomials in the $u_i$'s and $\overline{u_j}$'s are dense in $C^0(BV,\C)$, this obviously implies $\mu_V\equiv\mu_V'$ and the uniqueness of $\mu$ follows.
\end{proof}

We now turn to the 

\begin{proof}[Proof of Theorem~\ref{thm:other-localization}]
It suffices to prove the result for $f(\lambda)=\lambda^p$ with $p=0,1,...$. For $p=0$, this is nothing but the fact that $\sum_{k=0}^N\tr_{\gH^k}G_{N,k}^A=1$ since $G_N^A$ is a state. For $p=1$, we note that
$$\sum_{k=0}^N\frac{k}{N}\tr_{\gH^k}G_{N,k}^A=\tr(A\gamma_{\Gamma_N}^{(1)}A)$$
as follows from \eqref{eq:localized-DM}.
The right side can be written 
$$\tr(A\gamma_{\Gamma_N}^{(1)}A)=\tr(A(T+1)^{-1/2}(T+1)^{1/2}\gamma_{\Gamma_N}^{(1)}(T+1)^{1/2}(T+1)^{-1/2}A)$$
and this converges to $\tr (A\gamma^{(1)}A)$ since $A(T+1)^{-1/2}$ is compact and $(T+1)^{1/2}\gamma_{\Gamma_N}^{(1)}(T+1)^{1/2}\wto_*(T+1)^{1/2}\gamma^{(1)}(T+1)^{1/2}$ weakly-$\ast$ in $\gS^1$. So we obtain
\bqq
\lim_{N\to\ii}\sum_{k=0}^N\frac{k}{N}\tr_{\gH^k}G_{N,k}^A=\tr(A\gamma^{(1)}A)&=& \tr A\left(\int_{B\gH}d\mu(u)\, |u\rangle\langle u|\right)A \hfill\\
&=&\int_{B\gH}d\mu(u)\, \norm{Au}^2.
\eqq
The proof is similar for $p\geq2$, using that
$${N\choose p}^{-1}\sum_{k=p}^N{k\choose p}\tr_{\gH^k}G_{N,k}^A=\tr(A^{\otimes p}\gamma_{\Gamma_N}^{(p)}A^{\otimes p})$$
and the estimate~\eqref{eq:difference_k_N}.
\end{proof}

\section{Validity of Hartree's theory for trapped bosons}
\label{sec:trapped}

As a direct application of the usual de Finetti Theorem~\ref{thm:DeFinetti}, we state here a rather general result for trapped systems, which is in the spirit of the earlier works \cite{FanSpoVer-80,PetRagVer-89,RagWer-89}. 

Let $\gH$ be a separable Hilbert space and consider an $N$-body Hamiltonian of the form
\begin{equation}
H_N:=\sum_{j=1}^N T_j+\frac{1}{N-1}\sum_{1\leq k<\ell\leq N}w_{k\ell},
\label{eq:def_Hamiltonien_trapped} 
\end{equation}
acting on the bosonic space $\gH^N=\bigotimes_{s}^N \gH$. We make the assumption that 
\begin{equation}
\text{\it $T$ is bounded from below and it has a compact resolvent}
\label{eq:ass_T} 
\end{equation}
which is the mathematical formulation of the system being trapped.
We also assume that $w$ is a self-adjoint operator on $\gH^2$ which is relatively small compared to the one-particle term, in the sense that
\begin{equation}
-\beta_-(T\otimes 1+1\otimes T)-C\leq w\leq \beta_+ (T\otimes 1+1\otimes T)+C
\label{eq:ass_w}
\end{equation}
where $C\geq 0$, $\beta_+\ge 0$ and $1 >\beta_-\ge 0$. 
The upper bound in~\eqref{eq:ass_w} is not essential but it simplifies the analysis. Also, the precise form which we have chosen for $H_N$ is not really important. We could as well introduce $3$-body terms, or use abstract conditions like the ones in~\cite{RagWer-89,Werner-92}.

Under the assumption \eqref{eq:ass_w}, we deduce that
\begin{equation} 
(1-\beta_-)\sum_{j=1}^N T_j -CN \leq H_N \leq (1+\beta_+)\sum_{j=1}^N T_j + CN
\label{eq:upper_lower_bound_H_N}.
\end{equation}
This proves that $H_N$ is bounded from below and that its Friedrichs extension has the same form domain as the one-particle term $\sum_{j=1}^N T_j$. 

\subsection{Zero temperature case}
At zero temperature, we are interested in the ground state energy of $H_N$, which is given by
\begin{equation}
E(N)=\inf\sigma_{\gH^N}(H_N).
\end{equation}
The assumption that $T$ has a compact resolvent easily implies with~\eqref{eq:upper_lower_bound_H_N} that the spectrum of $H_N$ is purely discrete, hence $E(N)$ is an eigenvalue of finite multiplicity.

The corresponding Hartree problem is given by
\begin{equation}
\eH:=\inf_{\substack{u\in Q(T)\\ \norm{u}=1}}\cEH(u) := \inf_{\substack{u\in Q(T)\\ \norm{u}=1}} \left\{ \pscal{u,Tu}+\frac12\pscal{u\otimes u,w\, u\otimes u}_{\gH^2} \right\}
\label{eq:def_eH} 
\end{equation}
where $Q(T)$ is the quadratic form domain of the bounded-below operator $T$. Again the compactness of the resolvent of $T$ and our assumptions on $w$ easily imply the existence of at least one minimizer for $\eH$. Also, the set of minimizers $\cM\subset S\gH$ is bounded in $Q(T)$.

\begin{theorem}[\textbf{Validity of Hartree and BEC for trapped bosons}]\label{thm:confined}\mbox{}\\
Under the previous assumptions~\eqref{eq:ass_T} and~\eqref{eq:ass_w} on $T$ and $w$, we have 
$$\boxed{\lim_{N\to\ii}\frac{E(N)}{N}=\eH.}$$
If $(\Psi_N)$ is any sequence such that $\pscal{\Psi_N,H_N\Psi_N}=E(N)+o(N)$, then there exists a subsequence and a probability measure $\mu$ on the set $\cM$ of minimizers of $\E_{\rm H}$ (modulo a phase), such that 
$$\lim_{j\to\ii} \gamma^{(k)}_{\Psi_{N_j}}=\int_{\cM}d\mu(u)\;|u^{\otimes k}\rangle\langle u^{\otimes k}|$$
strongly in the trace-class for any fixed $k$ and, even, for the norm induced by the quadratic form of $T$:
\begin{equation}
\lim_{j\to\ii}  \norm{ A_k^{1/2}\left(\gamma^{(k)}_{\Psi_{N_j}}-\int_{\cM}d\mu(u)\;|u^{\otimes k}\rangle\langle u^{\otimes k}|\right) A_k^{1/2}}_{\gS^1}  =0.
\label{eq:limit_quad_form} 
\end{equation}
where $A_k:=\sum_{\ell=1}^k (T_\ell+C)$. In particular, if $\eH$ admits a unique minimizer $u_0$, then there is \emph{complete Bose-Einstein condensation} on $u_0$:
\begin{equation}
\lim_{N\to\ii} \gamma^{(k)}_{\Psi_N}=|u_0^{\otimes k}\rangle\langle u_0^{\otimes k}|.
\label{eq:BEC-confined} 
\end{equation}
\end{theorem}

Theorem~\ref{thm:confined} applies to several typical situations.

\begin{example}[\textbf{Bounded domain}]
Take $\gH=L^2(\Omega)$ where $\Omega$ is a bounded domain of $\R^d$ and $T=-\Delta+V$ with chosen boundary conditions and $w_{k\ell}=w(x_k-x_\ell)$ with $V,w\in L^p(\Omega)$ for some $p>\max(1,d/2)$. Hartree's theory is then valid in the mean-field limit by Theorem~\ref{thm:confined}.

In one dimension, the interaction can also be a delta potential, $w_{k\ell} = \lambda \delta (x_k-x_{\ell})$, as in the Lieb-Liniger model \cite{LieLin-63} on a finite interval $\Omega = (0,1)$. This potential, which acts as 
\begin{equation}
\pscal{\Psi,\delta(x_1-x_2)\Psi}:=\int_{\Omega}|\Psi(x,x)|^2\,dx
\label{eq:def_delta} 
\end{equation}
on two-body wave functions $\Psi$, is relatively form-bounded with respect to the Laplacian similarly as in~\eqref{eq:ass_w}, by the Sobolev embedding. By Theorem~\ref{thm:confined}, the associated Hartree functional 
\[
\cEH (u) = \int_\Omega \left(|u'(x)| ^2 +V(x)|u(x)|^2+ \frac{\lambda}{2} |u(x)| ^4 \right)dx
\]
gives the leading order of the energy per particle when $N\to \infty$, for any fixed $\lambda\in\R$. For $\lambda \geq 0$ this has been proved before in \cite{SeiYngZag-12}.
\end{example}

\begin{example}[\textbf{Confining potential}]
Take $\gH=L^2(\R^d)$, $T=-\Delta+V(x)$ with $V\in L^p_{\rm loc}(\R^d)$ satisfying $\lim_{|x|\to \infty}V(x)=+\ii$, and $w_{k\ell}=w(x_k-x_\ell)$ with $w\in L^p_{\rm loc}(\R^d)$ satisfying
\begin{equation*}
-\alpha (V(x)+V(y))- C \le w(x-y) \mathds{1} (|x-y|\ge R) \le \beta (V(x)+V(y))+C
\end{equation*}
for some constants $p>\max(1,d/2)$, $C\ge 0$, $R\ge 0$, $\beta\ge 0$, $1>\alpha \ge 0$. We see that the above assumptions~\eqref{eq:ass_T} and~\eqref{eq:ass_w} are satisfied, and therefore Hartree's theory is again valid by Theorem~\ref{thm:confined}. When $\widehat{w}\geq0$, this was proved in~\cite{GreSei-13}, where the energy $E(N)$ is even expanded to the next order.
\end{example}

\begin{example}[\textbf{Magnetic fields, rotation}]\label{ex:magn rot}
Take $\gH=L^2(\R^d)$, $T=-(\nabla+iA(x)) ^2 + V(x)$ with $V$ and $w$ as in the previous example, and $A\in L^p_{\rm loc}(\R^d)$. Hartree's theory is again valid in this case by Theorem~\ref{thm:confined}. One may for example consider $A(x) = (B/2) (x_2,- x_1,0)$ the vector potential corresponding to a uniform magnetic field $\mathrm{curl} A$ of strength $B$ pointing in the $x_3$ direction. Using the analogy between the Coriolis and the Lorentz force, one may also think of $A(x) = (\Omega/2) (x_2,- x_1,0)$ with $V = \tilde{V} - \frac{1}{4}|A| ^2$, which corresponds to a gas rotating at speed $\Omega$ around the $x_3$ axis, in the trapping potential $\tilde{V}$. The gas is then described in the rotating frame and the potential $V$ (called the effective potential) takes into account the effect of the centrifugal force. 
\end{example}

\begin{example}[\textbf{Bosons in the lowest Landau level}]
It is possible to deal with 3D bosons rotating along a fixed axis, in a harmonic confining potential and with a delta interaction (defined similarly as in~\eqref{eq:def_delta}), provided that the Hamiltonian is restricted to the lowest Landau level. See for instance~\cite{LieSeiYng-09} and the references therein for a precise definition of the model. In the Bargmann representation, the one-particle operator becomes $T=z\partial_z$ where $z\in\C$ is the complex coordinate in the plane orthogonal to the axis of rotation. The delta potential is itself a bounded operator on the Bargmann space. Theorem~\ref{thm:confined} applies to both the repulsive and attractive cases. The former was treated in~\cite{LieSeiYng-09}.
\end{example}

We can now provide the

\begin{proof}[Proof of Theorem~\ref{thm:confined}]
Let $(\Psi_N)$ be such that $\pscal{\Psi_N,H_N\Psi_N}=E(N)+o(N)$. Up to a subsequence we may assume that $\gamma_{\Psi_N}^{(k)} \wto \gamma^{(k)}$ weakly-$\ast$ in the trace class as $N\to \infty$ for all $k\ge 1$. By (\ref{eq:upper_lower_bound_H_N}) we know that $\Tr (T \gamma_{\Psi_N}^{(1)})$ is bounded and we therefore also get that $(T+C_0)^{1/2} \gamma_{\Psi_N}^{(1)} (T+C_0)^{1/2} \wto (T+C_0)^{1/2} \gamma^{(1)} (T+C_0)^{1/2}$ weakly-$\ast$  in the trace class, where $C_0$ is chosen such that $T+C_0\ge 1$. From the compactness of $(T+C_0)^{-1}$ it follows that
\begin{equation*}
\tr\gamma^{(1)}_{\Psi_N}=\tr \left(T+C_0 \right)^{-1} \left(T+C_0 \right)^{1/2} \gamma^{(1)}_{\Psi_N} \left(T+C_0 \right)^{1/2}  \to \tr\gamma^{(1)}.
\end{equation*}
By Corollary~\ref{cor:strong_CV}, we infer that $\gamma^{(k)}_{\Psi_N}\to\gamma^{(k)}$ strongly in $\gS^1(\gH^k)$ for all $k\geq1$. By Theorem~\ref{thm:DeFinetti}, there exists a probability measure $\mu$ on $Q(T)\cap S\gH$ such that 
$$\gamma^{(k)}=\int_{S\gH}d\mu(u)\;|u^{\otimes k}\rangle\langle u^{\otimes k}|$$
for all $k\geq1$. Now we go back to the $N$-body energy and notice that 
\begin{align}
\frac{\pscal{\Psi_N,H_N\Psi_N}}{N}&=\tr_{\gH} \left(T\gamma^{(1)}_{\Psi_N}\right)+\frac12\tr_{\gH^2} \left(w\gamma^{(2)}_{\Psi_N}\right)\nonumber\\
&=(\alpha-2)C_0+(1-\alpha)\tr_{\gH} \left( (T+C_0)\gamma^{(1)}_{\Psi_N}\right) \nonumber\\
&\quad+\frac12 \tr_{\gH^2} \left[ \left(w+\alpha(T\otimes 1 + 1 \otimes T)+2C_0 \right) \gamma^{(2)}_{\Psi_N} \right]\label{eq:trick_pass_limit}
\end{align}
where the constant $C_0$ is now chosen such that $T+C_0 \ge 1$ and $w+\alpha(T\otimes 1 + 1 \otimes T)+2C_0 \ge 0$. By Fatou's lemma, we obtain 
\begin{equation*}
\lim_{N\to\ii}\frac{\pscal{\Psi_N,H_N\Psi_N}}{N}\geq \tr_{\gH} \left(T\gamma^{(1)}\right)+\frac12\tr_{\gH^2} \left(w\gamma^{(2)}\right)
=\int_{S\gH}d\mu(u)\;\cEH(u)\geq \eH.
\end{equation*}
Recalling the upper bound $E(N)\le N e_{\rm H}$, we conclude that
$$\lim_{N\to\ii}\frac{\pscal{\Psi_N,H_N\Psi_N}}{N}=\int_{S\gH}d\mu(u)\;\cEH(u)=\eH.$$
It is now clear that $\mu$ must have its support on the set $\cM$ containing all the minimizers of $\eH$. 
Going back to the proof we see that equality must hold everywhere, which implies in particular that
$\tr(T\gamma^{(1)}_{\Psi_N}) \to \tr (T\gamma^{(1)}). $
By the reciprocal of Fatou's lemma, this gives the strong convergence 
$$(T+C_0)^{1/2} (\gamma_{\Psi_N}^{(1)} -\gamma^{(1)})(T+C_0)^{1/2}\to 0$$
in the trace class. To get the same for the higher density matrices, we use the same argument and the fact that  $\tr_{\gH} \left(T\gamma^{(1)}_{\Psi_N}\right)=k^{-1}\tr_{\gH^k} \left(\big(\sum_{i=1}^kT_i\big)\gamma^{(k)}_{\Psi_N}\right)$.

We mention that everything holds here for a subsequence. However, if $\cM=\{u_0\}$ then our proof shows that all the subsequences must have the same limit and the corresponding convergence follows for the whole sequence $\gamma_{ \Psi_N}^{(k)}$.
\end{proof}

\subsection{Positive temperature case}

Our result can be extended to the case of positive temperature and we quickly explain this now. We assume that 
\begin{equation}
\tr e^{-\beta T}<\ii 
\label{eq:ass_T_positive_temp}
\end{equation}
for all $\beta>0$ (which implies that the resolvent $(T+C)^{-1}$ must be trace-class, hence compact). The condition can be relaxed by assuming only $\tr e^{-\beta T}<\ii$ for \emph{one} $\beta>0$, but we use the stronger assumption~\eqref{eq:ass_T_positive_temp} for simplicity.

The $N$-particle quantum free energy at temperature $\beta^{-1}\geq0$ is now given by the formula
\begin{equation}
E(\beta,N)=-\frac{1}{\beta}\log\left(\tr_{\gH^N}e^{-\beta H_N}\right).
\end{equation}
From~\eqref{eq:upper_lower_bound_H_N} we have
$$\tr_{\gH^N}e^{-\beta H_N}\leq e^{\beta CN}\tr_{\gH^N}e^{-\beta (1-\alpha)\sum_{j=1}^N T_j}\leq e^{\beta CN} \left(\tr_\gH e^{-\beta (1-\alpha)T}\right)^N$$
which proves that
\begin{equation}
E(\beta,N)\geq -N\left(C+\frac{1}{\beta}\log\tr_\gH e^{-\beta (1-\alpha)T}\right).
\label{eq:lower_bound_positive_temp} 
\end{equation}
We see that, under the assumptions~\eqref{eq:ass_T_positive_temp} on $T$ and~\eqref{eq:ass_w} on $w$, the $N$-particle canonical free energy $E(\beta,N)$ is well defined for all $\beta>0$.
Let us denote by
\begin{equation} \label{eq:Gibb-state}
\Gamma_{\beta,N}=\frac{e^{-\beta H_N}}{\tr_{\gH^N}e^{-\beta H_N}}.
\end{equation}
the associated Gibbs state and by $\gamma^{(k)}_{\beta,N}=\tr_{k+1\to N}\Gamma_{\beta,N}$ its $k$-particle density matrices. Our result is then the following

\begin{theorem}[\textbf{Trapped bosons at positive temperature}]\label{thm:confined_positive_temp}\mbox{}\\
Under the previous assumptions~\eqref{eq:ass_T_positive_temp} and~\eqref{eq:ass_w} on $T$ and $w$, we have 
$$\boxed{\lim_{N\to\ii}\frac{E(\beta,N)}{N}=\eH}$$
for any fixed $\beta>0$. The other conclusions of Theorem~\ref{thm:confined} remain valid for the $k$-particle density matrices $\gamma^{(k)}_{\beta,N}$.
\end{theorem}

Remark that, at a fixed temperature $\beta^{-1}>0$, almost of the particles still condense on the Hartree ground state. In this case, the effect of the temperature is only seen on the next order of the free energy, given by Bogoliubov's theory (see, for example,~\cite[Prop. 3.1]{Suto-03} for non-interacting systems and~\cite{LewNamSerSol-13} for interacting systems). One can see the effect of the temperature in Hartree theory by taking a temperature which diverges proportionally to $N$, that is $\beta_N\simeq 1/N$, see~\cite{Gottlieb-05}.

\begin{proof}
We only have to show that $E(\beta,N)=E(N)+o(N)$ and the result then follows immediately from Theorem~\ref{thm:confined}.
First, by Gibbs' variational principle, we have
$$E(\beta,N)=\inf_{\substack{0\leq\Gamma:\gH^N\to\gH^N\\ \tr\Gamma=1}}\left(\tr_{\gH^N}H_N\Gamma+\beta^{-1}\tr_{\gH^N}\Gamma\log\Gamma\right)\leq E(N),$$
where the upper bound is obtained by taking $\Gamma=|\Psi_N\rangle\langle\Psi_N|$ with $\Psi_N$ the zero-temperature ground state. 

For the lower bound, we consider $\Gamma_{\beta,N}$ as in (\ref{eq:Gibb-state}). By arguing similarly as for the proof of~\eqref{eq:lower_bound_positive_temp}, it is easy to verify that
$$\tr_{\gH^N}\left(\sum_{i=1}^NT_i\right)\Gamma_{\beta,N}\leq CN$$
for a constant which is independent of $N$. Using this bound we can write 
\begin{equation*}
E(\beta,N)\geq E(N)+\epsilon \tr_{\gH^N}\left(\sum_{i=1}^NT_i\right)\Gamma_{\beta,N}+\beta^{-1}\tr_{\gH^N}\Gamma_{\beta,N}\log\Gamma_{\beta,N}-\epsilon CN.
\end{equation*}
where we have used that $H_N\geq E(N)$, the ground state energy at zero temperature.
In order to control the error terms, we use the exact behavior of the free energy in the non-interacting case.

\begin{lemma}[\textbf{A uniform lower bound in the non-interacting case}]\label{lem:lower_bound_uniform}\mbox{}\\
Let $K$ be an operator on $\gH$ such that $\tr_\gH e^{-\beta K}<\ii$ and let $\lambda_1$ be the first eigenvalue of $K$ which we assume to be non-degenerate. Then we have
\begin{equation}
0\leq -\log\left(\tr_{\gH^N}e^{-\sum_{i=1}^NK_i}\right) - N\lambda_1(K)+ \log\left(\tr_{\gH_\perp} \frac{1}{e^{- (K-\lambda_1)}-1}\right)\underset{N\to\ii}{\longrightarrow}0
\label{eq:lower_uniform}
\end{equation}
where $\gH_\perp={\mathds 1}_{(\lambda_1,\ii)}(K)\gH\subset \gH$.
\end{lemma}

The proof of the lemma is a well-known simple calculation which can be found, for instance, in~\cite[Prop. 3.1]{Suto-03}. There is a similar statement (with a different lower bound) when the first eigenvalue is degenerate, which we will not use here. 

In our case we do not know the multiplicity of the first eigenvalue of $T$, but we can argue as follows. First we add a positive finite-rank operator $B$ to $T$ in order to remove the degeneracy of $\lambda_1(T)$, without changing the first eigenvalue. We obtain an operator $T'=T+B$ to which we can apply Lemma~\ref{lem:lower_bound_uniform}. The error made by replacing $T$ by $T'$ is $\epsilon\norm{B}N$, leading to the following estimate:
\begin{align*}
&\epsilon \tr_{\gH^N}\left(\sum_{i=1}^NT_i\right)\Gamma_{\beta,N}+\beta^{-1}\tr_{\gH^N}\Gamma_{\beta,N}\log\Gamma_{\beta,N}\\
&\qquad\qquad \geq \epsilon \tr_{\gH^N}\left(\sum_{i=1}^NT'_i\right)\Gamma_{\beta,N}+\beta^{-1}\tr_{\gH^N}\Gamma_{\beta,N}\log\Gamma_{\beta,N}-\epsilon \norm{B}N\\
&\qquad\qquad \geq N\epsilon\lambda_1(T)-\frac1\beta \log\left(\tr_{\gH_\perp} \frac{1}{e^{-\beta\epsilon (T+B-\lambda_1(T))}-1}\right)
-\epsilon \norm{B}N.
\end{align*}
By choosing $\epsilon\to0$ slowly enough, we obtain that $E(\beta,N)\geq E(N)+o(N)$, and the rest follows from Theorem~\ref{thm:confined}.
\end{proof}


\section{Validity of Hartree's theory for unconfined bosons}\label{sec:non-trapped}

In this section we deal with quantum systems in which particles are allowed to escape to infinity. We will prove Theorem~\ref{thm:general_intro} which was stated in the introduction.

\subsection{Preliminaries}\label{sec:model}
We consider the Hamiltonian
\begin{equation}
H_N^V:=\sum_{j=1}^N\Big((m^2-\Delta_{j})^{s}-m^{2s}+V(x_j)\Big)+\frac{1}{N-1}\sum_{1\leq k<\ell\leq N}w(x_k-x_\ell)
\label{eq:HN-general}
\end{equation}
on the $N$-particle bosonic space $\gH^N=\bigotimes_s^NL^2(\R^d)$, where $m\ge 0$ and $s\in (0,1]$ are given constants. It will be convenient to introduce the notation
$$K=(m^2-\Delta)^{s}-m^{2s}\qquad\text{and}\qquad T=(m^2-\Delta)^{s}-m^{2s}+V.$$

For the rest of the paper, we always work with Assumptions (\ref{eq:assumption-V-w-1})--(\ref{eq:assumption-V-w-2})--(\ref{eq:assumption-V-w-3}) on $V$ and $w$. Under these assumptions,  we have
\begin{equation}
(1-\alpha_--\beta_-)\sum_{j=1}^NK_j-CN\leq H_N^V\leq (1+\alpha_++\beta_+)\sum_{j=1}^NK_j+CN
\label{eq:lower_bd_H_N_V}
\end{equation}
for some constant $C$ which could be different from the one of~\eqref{eq:assumption-V-w-3}. In particular, the quadratic form associated with $H^V_N$ has the same domain as that of the free kinetic energy $\sum_{j=1}^NK_j$. We denote by $E^V(N)$ the bottom of the spectrum of $H^V_N$. The following says that $E^V(N)N^{-1}$ always has a limit.

\begin{lemma}[\textbf{Monotonicity of the energy per particle}]\label{lem:increasing}\mbox{}\\
The sequence $N\mapsto E^V(N)N^{-1}$ is increasing and $\leq0$.
\end{lemma}

\begin{proof}
We have already explained in the introduction that the energy per particle is increasing because of the monotonicity of the sets $\cP^{(k)}_N$. It remains to show that $E^V(N)\leq0$ for all $N\geq1$. By taking a Hartree state $u^{\otimes N}$ with $u\in H^s(\R^d)$, we find
\begin{equation*}
\frac{E^V(N)}N\leq \pscal{u,Tu}+\frac12\iint w(x-y)|u(x)|^2|u(y)|^2\,dx\,dy. 
\end{equation*}
We fix a smooth normalized function of compact support $v\in C^\ii_c(\R^d)$ and take $u=k^{-1/2}\sum_{j=1}^kv(\cdot-x_n^j)$. In the limit where $|x_n^j|\to\ii$ as $n\to\ii$ and $|x^j_n-x^{j'}_n|\to\ii$ for $j\neq j'$, the energy of $u$ becomes the sum of the energies of the pieces, leading to
$$\frac{E^V(N)}N\leq \pscal{v,Kv}+\frac1{2k}\iint w(x-y)|v(x)|^2|v(y)|^2\,dx\,dy.$$
We have used here that
$$\lim_{|z|\to\ii}\iint w(x-y)|v(x)|^2|v(y+z)|^2\,dx\,dy=0$$
for any fixed function $v\in C^\ii_c(\R^d)$, which follows from our assumption~\eqref{eq:assumption-V-w-2} on $w$ at infinity.
Taking now $k\to\ii$ leads to
$E^V(N)/N\leq\inf\sigma((m^2-\Delta)^{s/2}-m^s)=0$,
which concludes the proof.
\end{proof}

Now we turn to the elementary properties of Hartree functional 
\begin{equation*}
\cEH^V(u)=\pscal{u,\left(\big(m^2-\Delta)^{s/2}-m^s+V\right)u}+\frac12\iint w(x-y)|u(x)|^2|u(y)|^2\,dx\,dy.
\end{equation*}
By~\eqref{eq:lower_bd_H_N_V}, it is bounded from below on bounded subsets of $L^2(\R^d)$. We denote its infimum on the sphere of radius $\sqrt{\lambda}$ by
$$\eH^V(\lambda):=\inf_{\substack{u\in H^s(\R^d)\\ \norm{u}^2=\lambda}}\cEH^V(u).$$
Some basic properties of Hartree's theory are given in the following lemma.

\begin{lemma}[\textbf{Hartree theory}]\label{le:Hartree}\mbox{}\\
Under Assumptions \eqref{eq:assumption-V-w-1}, \eqref{eq:assumption-V-w-2} and \eqref{eq:assumption-V-w-3}, we have 
\begin{equation}
\eH^V(1)\leq \eH^V(\lambda) + e_{\rm H}^0 (1-\lambda) \le \eH^0(\lambda)\leq0 
\label{eq:large_binding_Hartree}
\end{equation}
for all $\lambda\in [0,1]$. Moreover, if $$\eH^V(1)< \eH^V(\lambda) + e_{\rm H}^0 (1-\lambda)$$
for all $\lambda\in [0,1)$, then all the minimizing sequences for $e_{\rm H}(1)$ are relatively compact in $H^{s}(\R^d)$ and converge, after extraction of a subsequence, to a minimizer in $S\gH \cap Q(T)$. 
\end{lemma}

\begin{proof}
That $\eH^V(\lambda)\leq \eH^0(\lambda)\leq0$ was already shown in the proof of Lemma~\ref{lem:increasing}. 
The rest follows from the concentration-compactness method and the localization procedure explained in the appendix of~\cite{LenLew-11}. The proof will be omitted.
\end{proof}

In the next subsections we prove the validity of Hartree's theory for the many-particle operator $H^V_N$. We first deal with repulsive systems, then consider the purely translation-invariant case $V\equiv0$, before we turn to the general case.

\subsection{Weakly lower semi-continuous systems}\label{sec:wlsc}

We deal here with the case in which the particles escaping to infinity cannot form bound states. This is formalized by saying that $E^0(N)=0$ for all $N\geq2$ which, in the mean-field regime studied in this paper, turns out to be equivalent to $E^0(2)=0$.

\begin{lemma}[\textbf{Absence of bound states at infinity}]\label{lem:no-bound-state}\mbox{}\\
The following are equivalent:
\begin{enumerate}
\item[$(i)$] the operator $(m^2-\Delta)^{s}-m^{2s}+w/2$ has no negative eigenvalue on $L^2(\R^d)$;
\item[$(ii)$] $E^0(2)=0$;
\item[$(iii)$] $E^0(N)=0$ for all $N\geq1$.
\end{enumerate}
\end{lemma}
\begin{proof}
Since $E^0(2)$ is the bottom of the spectrum of $H^0_2$, $E^0(2)=0$ is clearly equivalent to $H^0_2=(m^2-\Delta_x)^{s}+(m^2-\Delta_y)^{s}-2m^{2s}+w(x-y)\geq0$. By removing the center of mass similarly as in~\cite[Appendix A]{LewSieVug-97}, this is the same as $(m^2-\Delta)^{s}-m^{2s}+w/2\geq0$. Finally, since $N\mapsto E^0(N)N^{-1}$ is non-decreasing and $\leq0$ by Lemma~\ref{lem:increasing}, $E^0(2)=0$ clearly implies $E^0(N)=0$ for all $N\geq2$.
\end{proof}

Now, the following says that, under one of the equivalent assumptions of Lemma~\ref{lem:no-bound-state}, the energy is a weakly  lower semi-continuous function of the one- and two-particle density matrices. 

\begin{proposition}[\textbf{Weak lower semi-continuity of the energy}]\label{prop:wlsc}\mbox{}\\
We assume that $V$ and $w$ satisfy \eqref{eq:assumption-V-w-1}--\eqref{eq:assumption-V-w-2}--\eqref{eq:assumption-V-w-3}, and that $w$ satisfies one of the equivalent assumptions of Lemma~\ref{lem:no-bound-state}. Then the energy of $H^V_N$ is a weakly lower semi-continuous function of the first two density matrices: If we have two sequences $\gamma^{(1)}_N,\gamma^{(2)}_N\geq0$ with $\tr_{\gH^2}\gamma^{(2)}_N=1$ and $\gamma_N^{(1)}=\Tr_{2}\gamma^{(2)}_N$, such that $\gamma^{(k)}_N\wto_*\gamma^{(k)}$ weakly-$\ast$ in $\gS^1(\gH^k)$ for $k=1,2$, then
\begin{equation}
\liminf_{N\to \infty} \left( \Tr_\gH[T \gamma_N^{(1)}]+\frac{1}{2}\Tr_{\gH^2}[w \gamma_N^{(2)}] \right) \ge \Tr_\gH[T \gamma^{(1)}]+\frac{1}{2}\Tr_{\gH^2}[w \gamma^{(2)}].
\label{eq:wlsc} 
\end{equation}
\end{proposition}

By using Proposition~\ref{prop:wlsc} we can now easily prove the main theorem for such systems.

\begin{theorem}[\textbf{Validity of Hartree's theory in the wlsc case}]\label{thm:wlsc}\mbox{}\\
We assume that $V$ and $w$ satisfy \eqref{eq:assumption-V-w-1}--\eqref{eq:assumption-V-w-2}--\eqref{eq:assumption-V-w-3}, and that $w$ satisfies one of the equivalent assumptions of Lemma~\ref{lem:no-bound-state}.

\medskip

\noindent $(i)$ We always have
\begin{equation}
\boxed{\lim_{N\to\ii}\frac{E^V(N)}{N}=\eH^V(1).}
\label{eq:limit_energy_wlsc}
\end{equation}

\medskip

\noindent $(ii)$ Denote by $\Psi_N$ a sequence of (normalized) approximate ground states in $\gH^N$, that is, such that $\pscal{\Psi_N,H_N^V\Psi_N}= E^V(N)+o(N)$, and by $\gamma^{(k)}_N$ the corresponding density matrices. Then there exists a subsequence $(N_j)$ and a Borel probability measure $\mu$ on the unit ball $B\gH$, supported on the set $\cM^V=\{u\in B\gH\ :\ \cEH^V(u)=\eH^V(\norm{u}^2)=\eH^V(1)\}$,
such that 
\begin{equation}
\boxed{\gamma^{(k)}_{N_j}\wto_*\int_{\cM^V}|u^{\otimes k}\rangle\langle u^{\otimes k}|\,d\mu(u)}
\label{eq:limit_weak}
\end{equation}
weakly-$\ast$ in $\gS^1(\gH^k)$, for all $k\geq1$.

\medskip

\noindent $(iii)$ Assume now that the binding inequality $\eH^V(1)<\eH^V(\lambda)$ is satisfied for all $0\leq\lambda<1$. Then the previous measure $\mu$ is supported on $S\gH$ and the limit~\eqref{eq:limit_weak} for $\gamma^{(k)}_{N_j}$ is strong in the trace-class. 
\end{theorem}

The proof of this theorem is an immediate consequence of Proposition~\ref{prop:wlsc} and of the weak de Finetti Theorem~\ref{thm:weak-De-Finetti}.

\begin{proof}
Let $\{\Psi_N\}$ be a sequence of approximate ground states as in the statement $(ii)$. Extracting subsequences if necessary, we may assume that $\gamma^{(k)}_N\wto_\ast\gamma^{(k)}$ for all $k\geq1$. By the weak de Finetti Theorem~\ref{thm:weak-De-Finetti}, there exists a Borel probability measure $\mu$ on $B\gH\cap H^s(\R^d)$ such that 
\bq \label{eq:conv-general}
\gamma^{(k)}=\int_{B\gH}d\mu(u)\;|u^{\otimes k}\rangle\langle u^{\otimes k}|
\eq
for all $k\ge 1$. From Proposition~\ref{prop:wlsc}, we have
\begin{multline*}
 \lim_{N\to \infty} \frac{E(N)}{N} = \lim_{N\to \infty} \left( \Tr_\gH[T \gamma_N^{(1)}]+\Tr_{\gH^2}[w \gamma_N^{(2)}] \right)\\
 \ge \Tr_\gH[T \gamma^{(1)}]+\Tr_{\gH^2}[w \gamma^{(2)}] = \int_{B\gH} d\mu(u) \cEH^V(u) \ge \int_{B\gH} d\mu(u) \eH^V(\|u\|^2).  
\end{multline*}
By~\eqref{eq:large_binding_Hartree} we have $\eH^V(\norm{u}^2)\geq \eH^V(1)$ and therefore the result follows from the fact that $\mu$ is a probability measure.
\end{proof}

\begin{remark}[An abstract result]
It is clear that the proof of Theorem~\ref{thm:wlsc} is general and there is a similar statement in a purely abstract setting. The only two properties which we have used are $(i)$ that the energy of $H^V_N$ is a weakly lower semi-continuous function of the first two density matrices; and $(ii)$ that $\eH^V(\lambda)\geq\eH^V(1)$ for all $0\leq \lambda\leq1$. The first means that particles escaping to infinity always carry a non-negative energy. On the other hand, the second property  means that, in Hartree theory, we can send some of the particles to infinity without any cost. The two assumptions are complementary and both are necessary here. By adding a constant to $T$, we can always ensure either $(i)$ or $(ii)$, but not both at the same time.
\end{remark}

Before proving Proposition~\ref{prop:wlsc} under the general assumptions of Lemma~\ref{lem:no-bound-state}, let us remark that this proof is obvious if $w\geq0$. In this case~\eqref{eq:wlsc} follows immediately from Fatou's lemma for operators and from the fact that the essential spectrum of $T$ starts at 0, that is, $\1(T\leq0)$ is a compact operator. An example of such purely repulsive systems is a bosonic atom.

\begin{example}[\textbf{Bosonic atoms}]\mbox{}\\
After scaling, an atom with a classical nucleus at the origin in $\R^3$ and $N$ ``bosonic electrons" is described by the Hamiltonian
\begin{equation} \label{eq:HN-bosonic-atom}
H_N:=\sum_{j=1}^N\left(-\Delta_i -\frac{1}{t|x_i|} \right)+\frac{1}{N-1}\sum_{1\leq k<\ell\leq N}\frac{1}{|x_k-x_\ell|},
\end{equation}
acting on $\gH^N = \bigotimes_{s}^N L^2(\R^3)$, where $t=(N-1)/Z$ is the ratio between the number of electrons and the nuclear charge. This Hamiltonian satisfies all our assumptions (\ref{eq:assumption-V-w-1})--(\ref{eq:assumption-V-w-2})--(\ref{eq:assumption-V-w-3}) with $d=3$, $s=1$, $V(x)=-(t|x|)^{-1}$ and $w(x)=|x|^{-1}$. Since $w\geq0$, the equivalent assumptions of Lemma~\ref{lem:no-bound-state} are also satisfied and~\eqref{eq:wlsc} follows immediately.

For any fixed $t>0$, the Hartree energy $e^V_{\rm H}(\lambda)$ is decreasing on $[0,\min(1,t_c/t]$ and constant on $[t_c/t,\ii)$, where $t_c\simeq 1.21$~\cite{BenLie-83,Baumgartner-84}. Furthermore, $\eH^V(\lambda)$ admits a unique minimizer $u_0$ if and only if $\lambda\leq t_c/t$. We deduce that the set $\cM^V$ appearing in the statement of Theorem~\ref{thm:wlsc} is
$$\cM^V=\{e^{i\theta} u_0\}_{\theta\in[0,2\pi)}\qquad\text{where}\qquad \int_{\R^3}|u_0|^2=\begin{cases}
1&\text{if $t\le t_c$,}\\
t_c/t&\text{if $t> t_c$,}\\
\end{cases}$$
where $u_0$ is the unique minimizer for $\eH^V(1)$ if $t\le t_c$ and for $\eH^V(t_c/t)$ if $t> t_c$.
Now we deduce from Theorem~\ref{thm:wlsc} that 
$$\gamma^{(k)}_N\begin{cases}
\to |u_0^{\otimes k}\rangle\langle u_0^{\otimes k}| & \text{strongly in $\gS^1(\gH^k)$, if $t\le t_c$,}\\[0.1cm]
\wto_\ast |u_0^{\otimes k}\rangle\langle u_0^{\otimes k}| & \text{weakly-$\ast$ in $\gS^1(\gH^k)$, if $t> t_c$,}\\
\end{cases}$$
for all $k\geq1$. To our knowledge, this is the first result about approximate minimizers for bosonic atoms beyond the critical value $t_c$.

The occurrence of complete Bose-Einstein condensation was shown by Benguria and Lieb in~\cite{BenLie-83} for $t<t_c$ (see also~\cite{Solovej-90,Bach-91,BacLewLieSie-93}). Their proof used the particular form of the Coulomb interaction through the Lieb-Oxford inequality. Recently, Kiessling also considered bosonic atoms in~\cite{Kiessling-12} for which he used the \emph{classical} de Finetti Theorem.
\end{example}

We end this section by providing the

\begin{proof}[Proof of Proposition~\ref{prop:wlsc}] If $\Tr_\gH[T\gamma_N^{(1)}]$ is not bounded, then there is nothing to prove. So we may assume $\Tr_\gH[T\gamma_N^{(1)}] \le C$. The proof is divided into two steps. 

\subsubsection*{\bf Step 1. Splitting of the energy.}
The first step is a classical result used many times in the literature, and which does not require that $w$ has no bound state at infinity. It is convenient to write $T=K+V$ with  $K=(m^2-\Delta)^{s/2}-m^s \ge 0.$

\begin{lemma}[\textbf{Splitting of the energy}]\label{lem:splitting}\mbox{}\\
Assume that \eqref{eq:assumption-V-w-1}, \eqref{eq:assumption-V-w-2} and \eqref{eq:assumption-V-w-3} are satisfied. Consider a smooth partition of unity $\chi^2+\eta^2=1$ where $\chi(x)=1$ if $|x|\le 1$ and $\chi(x)=0$ if $x\ge 2$, and define $\chi_R(x)=\chi(x/R)$ and $\eta_R=\eta(x/R)$. Let $\gamma^{(1)}_N$ and $\gamma^{(2)}_N$ be as in Proposition~\ref{prop:wlsc}. We have 
\begin{multline}
\label{eq:split-energy}
\liminf_{N\to \infty} \left( \Tr_\gH[T  
\gamma_N^{(1)}]+\frac{1}{2}\Tr_{\gH^2}[w \gamma_N^{(2)}] \right) \\
  \ge \liminf_{R\to \infty}\liminf_{N\to \infty} \left( \Tr_\gH[T  
\chi_R \gamma_N^{(1)} \chi_R ]+\frac{1}{2}\Tr_{\gH^2}[w  
\chi_R^{\otimes 2}\gamma_N^{(2)}\chi_R^{\otimes 2}] \right.\\
  \left. + \Tr_\gH[K \eta_R  
\gamma_N^{(1)}\eta_R]+\frac{1}{2}\Tr_{\gH^2}[w \eta_R^{\otimes  
2}\gamma_N^{(2)}\eta_R^{\otimes 2}] \right) .
\end{multline}
\end{lemma}

\begin{proof}[Proof of Lemma~\ref{lem:splitting}]
Again we may assume that $\Tr_\gH[K \gamma_N^{(1)}]\leq C$.
To split the kinetic energy, we use the estimate
\bq \label{eq:IMS-estimate}
\lim_{R\to \infty} \|K- \chi_R K \chi_R -\eta_R K \eta_R \|_{L^2\to L^2}=0
\eq
which follows for instance from the fractional IMS formula which can be found in~\cite[Lemma 7]{LenLew-11}. Moreover, using that $K^{1/2}\gamma_N^{(1)}K^{1/2}$ is bounded in the trace-class, we have
\begin{equation}
\left|\Tr(V\eta_R \gamma_N^{(1)}\eta_R)\right|\leq C\norm{(1-\Delta)^{-s/2}\eta_RV\eta_R(1-\Delta)^{-s/2}}_{L^2\to L^2}.
\label{eq:control_V_outside} 
\end{equation}
Our assumption (\ref{eq:assumption-V-w-2}) on $V$ implies that the norm on the right side goes to 0 when $R\to\ii$ and we obtain 
$$ \lim_{R\to \infty} \liminf_{N\to \infty} \Tr(V\eta_R \gamma_N^{(1)}\eta_R)=0.$$
This means that the particles at infinity will never see the potential $V$.

To split the two-body term, we insert $1=(\chi_R^2(x)+\eta_R^2(x))(\chi_R^2(y)+\eta_R^2(y))$ and expand. We only have to control the cross term, namely we have to prove that (possibly for a subsequence)  
\bq \label{eq:2body-cross}
\lim_{R\to \infty} \lim_{N\to \infty} \left(  \iint \chi_R^2(x) |w(x-y)| \eta_{R}^2(y) \rho^{(2)}_N(x,y)\,dx\,dy \right) =0
\eq
where $\rho_N^{(2)}(x,y)=\gamma_N^{(2)}(x,y;x,y)$. To this end we write $\eta_R^2(y)=\eta_R^2(y)-\eta_{4R}^2(y)+\eta_{4R}^2(y)$ and remark that $\chi_R^2(x) |w(x-y)| \eta_{4R}^2(y) \le \1_{\{|x-y|\ge R\}} w(x-y)$ which can be easily controlled using the same argument as for~\eqref{eq:control_V_outside}. So it remains to treat the term with $\chi_R^2(x)[\eta_R^2 (y) -\eta_{4R}^2(y) ]$, for which we claim that
\begin{multline*}
\lim_{R\to \infty}\lim_{k\to \infty} \left(  \iint |w(x-y)| \chi_R^2(x)[\eta_R^2 (y) -\eta_{4R}^2(y) ] \rho^{(2)}_{N_k}(x,y)\,dx\,dy \right)\\
\leq\lim_{R\to \infty}\lim_{k\to \infty} \left(  \iint |w(x-y)| \1(R\leq|y|\leq 8R)\rho^{(2)}_{N_k}(x,y)\,dx\,dy \right)=0 
\end{multline*}
for an appropriate subsequence $(N_k)$. This is an adaption of Lions' concentration-compactness argument. Let us introduce the concentration functions
$$ Q_{N}(R):= \iint |w(x-y)| \1(|y|\geq R) \rho^{(2)}_N(x,y)\,dx\,dy. $$
For every $N$, the function $R\mapsto Q_N(R)$ is non-increasing on $[0,\infty)$. Moreover, $0\le Q_N (R) \le \Tr_{\gH^2}[|w| \gamma_N^{(2)}] \le C_0$ by~\eqref{eq:assumption-V-w-3}. Therefore, by Helly's selection principle, there exists a subsequence $N_k$ and a decreasing function $Q:[0,\infty) \to [0,C_0]$ such that $Q_{N_k}(R)\to Q(R)$ for all $R\in [0,\infty)$. Since $\lim_{R\to \infty} Q(R)$ exists, we conclude that 
$$\lim_{R\to \infty} \lim_{k\to \infty} (Q_{N_k}(R) - Q_{N_k}(8R))= \lim_{R\to \infty} (Q(R)-  Q(8R) )=0.$$
This is the desired convergence.
\end{proof}

\subsubsection*{\bf Step 2. Passing to the limit.}
Now we look at the right side of (\ref{eq:split-energy}). 
By the local compactness we have
\bq \label{eq:wlsc-strong-conv-chi} \lim_{N\to \infty} \chi_R^{\otimes k} \gamma_N^{(k)} \chi_R^{\otimes k} = \chi_R^{\otimes k}\gamma_N^{(k)}\chi_R^{\otimes k}
\eq
strongly in the trace class for all $k\ge 1$ and any fixed $R$. Using the fact that for any one-body operator $0\leq A \leq 1$
\begin{equation}
A \gamma_N^{(1)}A - \tr_2 [A^{\otimes 2} \gamma_N^{(2)}A^{\otimes 2}] \ge 0~~\text{on}~\gH,
\label{eq:link_1pdm_2pdm} 
\end{equation}
(where $\Tr_2$ is the partial trace in the second variable) and Fatou's lemma, we get
\begin{multline}
\label{eq:wlsc-Fatou-1body}
\liminf_{N\to \infty} \Tr_\gH \left[(T+C_0) \left( \chi_R \gamma_N^{(1)}\chi_R - \tr_2 [\chi_R^{\otimes 2} \gamma_N^{(2)}\chi_R^{\otimes 2}] \right) \right]\\
\ge    \Tr_\gH \left[(T+C_0)\left( \chi_R \gamma^{(1)}\chi_R - \tr_2 [\chi_R^{\otimes 2} \gamma^{(2)}\chi_R^{\otimes 2}\ \right) \right]
\end{multline}
for a large enough $C_0$ chosen such that $T+C_0\geq0$. By Fatou's lemma again,
\begin{multline}
\label{eq:wlsc-Fatou-2body}
\liminf_{N\to \infty} \frac{1}{2} \Tr_{\gH^2} \left[(T\otimes 1 + 1\otimes T + 2C_0 +w ) \chi_R^{\otimes 2} \gamma_N^{(2)}\chi_R^{\otimes 2}] \right]\\
\ge \frac{1}{2}\Tr_{\gH^2} \left[(T\otimes 1 + 1\otimes T + 2C_0 +w ) \chi_R^{\otimes 2} \gamma^{(2)}\chi_R^{\otimes 2}] \right].
\end{multline}
Here we have to choose $C_0$ even larger to make sure that $(T\otimes 1 + 1\otimes T + 2C_0 +w \ge 0$ on $\gH^2$. Putting (\ref{eq:wlsc-Fatou-1body}) and (\ref{eq:wlsc-Fatou-2body}) together and using the strong convergence (\ref{eq:wlsc-strong-conv-chi}), we find that
\begin{multline}
\label{eq:wlsc-chi-limit-N}
\liminf_{N\to \infty} \left( \Tr_\gH[T\chi_R \gamma_N^{(1)}\chi_R] + \frac{1}{2} \Tr_{\gH^2} [w\chi_R^{\otimes 2} \gamma_N^{(2)}\chi_R^{\otimes 2}] \right)\\
\ge\Tr_\gH[T\chi_R \gamma^{(1)}\chi_R] + \frac{1}{2} \Tr_{\gH^2} [w\chi_R^{\otimes 2} \gamma^{(2)}\chi_R^{\otimes 2}].
\end{multline}
By taking now the limit $R\to \infty$ we arrive at
\begin{multline}
\label{eq:wlsc-chi-limit-N-R}
\liminf_{R\to \infty}\liminf_{N\to \infty} \left( \Tr_\gH[T\chi_R \gamma_N^{(1)}\chi_R] + \frac{1}{2} \Tr_{\gH^2} [w\chi_R^{\otimes 2} \gamma_N^{(2)}\chi_R^{\otimes 2}] \right)\\
\ge \Tr_\gH[T \gamma^{(1)}] + \frac{1}{2} \Tr_{\gH^2} [w\gamma^{(2)}].
\end{multline}

Up to now the argument is general and it holds without the assumption that $w$ has no bound state. We will only use this fact to estimate the second term on the right side of (\ref{eq:split-energy}), that is, to prove that the particles far away have a non negative energy. Indeed, using~\eqref{eq:link_1pdm_2pdm} and  the assumption that $H^0(2)\geq0$, we obtain
\begin{equation}
\label{eq:wlsc-eta-limit-N-R}
\Tr_\gH[K\eta_R \gamma_N^{(1)}\eta_R] + \frac{1}{2} \Tr_{\gH^2} [w\eta_R^{\otimes 2} \gamma_N^{(2)}\eta_R^{\otimes 2}]
\geq \frac{1}{2} \Tr_{\gH^2} \left[H^0_2 \eta_R^{\otimes 2}\gamma_N^{(2)}\eta_R^{\otimes 2}\right] \ge 0.
\end{equation}
The weak lower semi continuity  (\ref{eq:wlsc}) now follows from (\ref{eq:split-energy}), (\ref{eq:wlsc-chi-limit-N-R}) and (\ref{eq:wlsc-chi-limit-N-R}).
\end{proof}

\subsection{Translation invariant case} \label{sec:tr-in}

In this subsection, we ignore the potential $V$ and consider the ground state energy $E^0(N)$ of the fully translation-invariant Hamiltonian 
$$ H_N^0 = \sum_{i=1}^N ((m^2-\Delta_i)^{s}-m^{2s}) + \frac{1}{N-1} \sum_{i<j}^N w(x_i-x_j).$$
We work under the same assumptions (\ref{eq:assumption-V-w-1}), (\ref{eq:assumption-V-w-2}) and (\ref{eq:assumption-V-w-3}) as before, but do not assume any of the equivalent statements of Lemma \ref{lem:no-bound-state}.
The particles escaping to infinity can thus have a negative energy, and the simple proof of Theorem~\ref{thm:wlsc} in the previous section does not apply. Furthermore, the situation is complicated by the fact that we do not expect the local convergence of all the sequences of approximate ground states. By using the translation invariance and the kinetic energy of the center of mass, one can construct sequences for which $\gamma^{(1)}_{\Psi_N}(\cdot-x_N)\wto_\ast0$ for any translation $(x_N)\subset\R^d$. This is called \emph{vanishing} in Lions' terminology~\cite{Lions-84}. Our proof of the validity of Hartree's theory will be based on some ideas of Lieb, Thirring and Yau~\cite{LieThi-84,LieYau-87} and on the geometric techniques of~\cite{Lewin-11}. Our main result is the following.

\begin{theorem}[\textbf{Translation-invariant systems}]\label{thm:tr-in}\mbox{}\\
Under Assumptions \eqref{eq:assumption-V-w-1}, \eqref{eq:assumption-V-w-2} and \eqref{eq:assumption-V-w-3} on $w$, we have  
$$\boxed{ \lim_{N\to \infty} \frac{E^0(N)}{N} =e^0_{\rm H}(1).}$$
\end{theorem}

We recall that the ground state energy in Hartree theory is 
$$ e^0_{\rm H}(1):= \inf_{\norm{u}^2=1} \left\{ \langle u,K u\rangle + \frac{1}{2} \iint |u(x)|^2 w(x-y) |u(y)|^2 dxdy \right\}$$
where $K:= (m^2-\Delta)^{s}-m^{2s}$.

\medskip

\begin{example}[\textbf{Boson stars}]\mbox{}\\
A pseudo-relativistic model of a star with $N$ gravitating bosons may be described by the Hamiltonian $H_N^0$ in dimension $d=3$  with $m=1, s=1/2$ and $w(x)=-\kappa |x|^{-1}$ for some  $0<\kappa<2/\pi$. The validity of Hartree's theory in this case was proved by Lieb and Yau \cite{LieYau-87}. Their proof is based on a clever replacement of the two-body potential by a one-body potential, which is also crucial in our proof. 
\end{example}

The outline of our proof is as follows. Following Lieb and Yau~\cite{LieYau-87}, we use some part of the two-body potential to create a negative one-body potential which breaks the translation invariance of the system. Since $e_{\rm H}^0(1) \le 0$ by Lemma \ref{le:Hartree}, it suffices to consider the case when $\lim_{N\to \infty}{E(N)}/{N}<0$. In this case, we may rule out both the vanishing and dichotomy for the modified model, and hence the geometric method of~\cite{Lewin-11} applies. The final result then follows from an approximation argument.

\begin{proof}[Proof of Theorem \ref{thm:tr-in}] First, we remove the center of mass to create a negative one-body potential. To be precise, for any $N$-body wave function $\Psi\in \gH^N$ and $\eps>0$, we can write, similarly as in~\cite{LieThi-84,LieYau-87},
\begin{multline}
\frac{N-1}{N}\langle \Psi , H_N^0  \Psi \rangle\\ =\left\langle \Psi, \left( \sum_{i=1}^{N-1}\Big(K_i -\eps w_-(x_i-x_N) \Big) + \frac{1}{N-2} \sum_{i<j}^{N-1} w_\eps(x_i-x_j) \right) \Psi \right\rangle \label{eq:trick-Nam}
\end{multline}
where 
$w_-(x):=\max\{0,-w(x)\}$ and $w_\eps(x):=w(x)+2\eps w_-(x)$.
In~\eqref{eq:trick-Nam}, the Hamiltonian in the parenthesis depends on $x_N$ but, by translation-invariance of the other terms, the bottom of its spectrum is actually independent of $x_N$. Therefore, 
$$\frac{E^0(N)}{N}\ge \frac{E_\eps(N-1)}{N-1}$$
where $E_\eps(N-1)$ is the ground state energy of the modified Hamiltonian
$$ H_{\eps,N-1}= \sum_{i=1}^{N-1}\Big(K_i -\eps w_-(x_i) \Big) + \frac{1}{N-2} \sum_{i<j}^{N-1} w_\eps (x_i-x_j)$$
in $\gH^{N-1}$. Since $E^0(N)/N$ and $E_\eps(N)/N$ are increasing sequences by Lemma~\ref{lem:increasing}, the following limits exist:
$$a_\eps:=\lim_{N\to \infty}\frac{E_\eps(N)}{N} \le \lim_{N\to \infty} \frac{E^0(N)}{N}=:a \le e^0_{\rm H}(1)\le 0.$$

We may now assume that $a<0$; otherwise $a=e^0_{\rm H}(1)=0$ and the proof is finished. Note that 
$$\lim_{\eps\to 0}\inf \sigma(K-\eps w_-) = \inf \sigma(K)=0,$$
where we have used the lower bound $\alpha_- K+w \ge -C$ by~\eqref{eq:assumption-V-w-3}. Therefore, $\inf \sigma(K-\eps w_-)>a$ for $\eps>0$ small enough. In this case we shall show that $a_\eps$ is exactly equal to the modified Hartree energy  
\bqq e_{\rm H,\eps}(1) := \inf_{\norm{u}^2=1} \left\{ \langle u,(K-\eps w_-) u\rangle + \frac{1}{2} \iint |u(x)|^2 w_\eps (x-y) |u(y)|^2 dxdy \right\}.
\eqq
Since $a_\eps \to a$ and $e_{{\rm H},\eps}(1) \to e^0_{\rm H}(1)$ as $\eps\to 0$, we then conclude that $a=e^0_{\rm H}(1)$ as desired. 

In order to prove that $a_\eps=e_{{\rm H},\eps}(1)$, we consider a sequence $\Psi_N$ of wave functions such that $\langle \Psi_N, H_{\eps,N} \Psi_N \rangle = E_{\eps}(N)+o(N)$. Then
$$
a_\eps = \lim_{N\to \infty}\frac{\langle \Psi_N, H_{\eps,N} \Psi_N \rangle}{N}=  \lim_{N\to \infty} \left( \Tr_\gH[(K-\eps w_-)\gamma_N^{(1)}] + \frac{1}{2} \Tr_{\gH^2}[w_\eps \gamma_N^{(2)}] \right)
$$ 
where $\gamma_N^{(k)}$ are the $k$-particle density matrices of $\Psi_N$. Up to extraction of a subsequence we may assume that $\gamma^{(k)}_N\wto_\ast\gamma^{(k)}$ for all $k\geq1$. 
The equality $a_\eps = e_{\rm H,\eps}$ will follow immediately if we can show that
\bq \label{eq:concentration-case}
\tr[\gamma^{(1)}]=1.
\eq
Indeed, by Corollary~\ref{cor:strong_CV},~\eqref{eq:concentration-case} implies that $\gamma^{(k)}_N\to\gamma^{(k)}$ strongly for all $k\geq1$. Then, arguing like in~\eqref{eq:trick_pass_limit} for the proof of Theorem~\ref{thm:confined} in the confined case, one sees that, in this case,
\begin{multline*}
\liminf_{N\to\ii}\left( \Tr_\gH[(K-\eps w_-)\gamma_N^{(1)}] + \frac{1}{2} \Tr_{\gH^2}[w_\eps \gamma_N^{(2)}] \right)\\
\geq \Tr_\gH[(K-\eps w_-)\gamma^{(1)}] + \frac{1}{2} \Tr_{\gH^2}[w_\eps \gamma^{(2)}].
\end{multline*}
The result then follows from the (strong) quantum de Finetti Theorem~\ref{thm:DeFinetti}. 

So it only remains to prove our claim~\eqref{eq:concentration-case}. To this end we use the geometric localization method. We pick a smooth partition $\chi_R^2 +\eta_R^2=1$ like in Lemma~\ref{lem:splitting} and get 
\begin{multline}
\label{eq:split-energy-translation}
a_\eps \ge  \liminf_{R\to \infty}\liminf_{N\to \infty} \left\{ \Tr_\gH[(K-\eps w_-) \chi_R \gamma_N^{(1)}\chi_R] + \frac{1}{2} \Tr_{\gH^2}[w_\eps \chi_R^{\otimes 2}\gamma_N^{(2)}\chi_R^{\otimes 2}] \right. \\
\left. + \Tr_\gH[ K \eta_R \gamma_N^{(1)}\eta_R] + \frac{1}{2} \Tr_{\gH^2}[w_\eps \eta_R^{\otimes 2}\gamma_N^{(2)}\eta_R^{\otimes 2}] \right \} .
\end{multline}

We consider the $\chi_R$-- and $\eta_R$--localized states $G_N^\chi$ and $G_N^\eta$ of $\Psi_N$, defined in~\eqref{eq:def_localization} and~\eqref{eq:def_localization2}.
We have 
\begin{multline}
\label{eq:localization-chi-eps}
\Tr_\gH[(K-\eps w_-) \chi_R \gamma_N^{(1)}\chi_R] + \frac{1}{2} \Tr_{\gH^2}[w_\eps \chi_R^{\otimes 2}\gamma_N^{(2)}\chi_R^{\otimes 2}]\\
= \frac{1}{N} \sum_{k=1}^N\Tr_{\gH^k}\left[ \left( \sum_{i=1}^k (K-\eps w_-)_i + \frac{1}{N-1} \sum_{i<j}^k w_\eps (x_i-x_j)\right) G^\chi_{N,k} \right]
\end{multline}
and
\begin{multline}
\label{eq:localization-translation-eta}
\Tr_\gH[ T \eta_R \gamma_N^{(1)}\eta_R] + \frac{1}{2} \Tr_{\gH^2}[w_\eps \eta_R^{\otimes 2}\gamma_N^{(2)}\eta_R^{\otimes 2}]\\
= \frac{1}{N} \sum_{k=1}^N \Tr_{\gH^k}\left[ \left( \sum_{i=1}^k K_i + \frac{1}{N-1} \sum_{i<j}^k w_\eps (x_i-x_j)\right) {G}^{\eta}_{N,k} \right].
\end{multline}
We have to estimate these two terms from below.

\medskip

\paragraph*{\bf Estimate on~\eqref{eq:localization-chi-eps}.}
We apply the variational inequality 
\bq\label{eq:A-tB}
A+tB = (1-t)A+ t(A+B) \ge (1-t) \inf \sigma(A) + t \inf\sigma  (A+B) 
\eq
for $A=\sum_{\ell =1}^k (K-\eps w_-)_\ell$, $A+B=H_{\eps,k}$ and $t=(k-1)/(N-1)$. Note that in this case $\inf\sigma(A)\ge \inf \sigma(A+B)$ since 
$$\inf \sigma(K-\eps w_-) > a \ge a_\eps \ge k^{-1}\inf \sigma(H_{\eps,k}) .$$
Thus from (\ref{eq:localization-chi-eps}) it follows that  
$$\Tr_\gH[(K-\eps w_-) \chi_R \gamma_N^{(1)}\chi_R] + \frac{1}{2} \Tr_{\gH^2}[w_\eps \chi_R^{\otimes 2}\gamma_N^{(2)}\chi_R^{\otimes 2}]  \ge \sum_{k=1}^N  \frac{k\Tr G^{\chi}_{N,k}}{N} \; \frac{E_\eps (k)}{k}.$$
Since 
$$\sum_{k=0}^N \frac{k\Tr G^\chi_{N,k}}{N} = \Tr [\chi_R^2 \gamma_N^{(1)}]\underset{N\to\ii}{\longrightarrow}\Tr [\chi_R^2 \gamma^{(1)}]\quad \text{and}\quad \lim_{k\to \infty} \frac{E_\eps (k)}{k} =a_\eps,$$
we conclude that
\begin{multline}
\label{eq:energy-chiR-translation}
\liminf_{N\to \infty} \left( \Tr_\gH[(K-\eps w_-) \chi_R \gamma_N^{(1)}\chi_R] + \frac{1}{2} \Tr_{\gH^2}[w_\eps \chi_R^{\otimes 2}\gamma_N^{(2)}\chi_R^{\otimes 2}] \right)\\
\geq a_\eps \Tr [\chi_R^2 \gamma^{(1)}]. 
\end{multline}
 
\medskip

\paragraph*{\bf Estimate on~\eqref{eq:localization-translation-eta}.}
Using $K\ge 0$, $w_\eps = w+ 2\eps w_- \ge (1-2\eps)w$ and $E^0(k)\le ak<0$, we find that 
\bqq
\sum_{i=1}^k K_i + \frac{1}{N-1} \sum_{i<j}^k w_\eps (x_i-x_j) &\ge& \frac{(1-2\eps)(k-1)}{N-1} H^0_k \hfill\\
&\ge &  \frac{(1-2\eps)(k-1)}{N-1} E^0 (k) \ge E^0(k) - 2\eps a k
\eqq
for every $k\ge 1$. Thus from (\ref{eq:localization-translation-eta}) it follows that
\begin{equation*}
\Tr_\gH[ K \eta_R \gamma_N^{(1)}\eta_R] + \frac{1}{2} \Tr_{\gH^2}[w_\eps \eta_R^{\otimes 2}\gamma_N^{(2)}\eta_R^{\otimes 2}] \ge \sum_{k=1}^N  \frac{k\Tr G^{\eta}_{N,k}}{N} \cdot \left(\frac{E^0(k)}{k} - 2\eps a \right).
\end{equation*}
Since 
$$\sum_{k=0}^N \frac{k\Tr G^\eta_{N,k}}{N} = \Tr [\eta_R^2 \gamma_N^{(1)} ]\underset{N\to\ii}{\longrightarrow}1-\Tr [\chi_R^2 \gamma^{(1)}]\quad \text{and}\quad \lim_{k\to \infty} \frac{E^0(k)}{k}  =a,$$
we deduce that
\begin{multline}
\label{eq:energy-etaR-translation}
\liminf_{N\to \infty} \left( \Tr_\gH[ T \eta_R \gamma_N^{(1)}\eta_R] + \frac{1}{2} \Tr_{\gH^2}[w_\eps \eta_R^{\otimes 2}\gamma_N^{(2)}\eta_R^{\otimes 2}]\right)\geq (1-2 \eps) a \big(1-\Tr[\chi_R^2 \gamma^{(1)}]\big). 
\end{multline}

Substituting now (\ref{eq:energy-chiR-translation}) and (\ref{eq:energy-etaR-translation}) into (\ref{eq:split-energy-translation}), we find that
\begin{align*}
a_\eps &\ge  \liminf_{R\to \infty}\left( a_\eps \Tr [\chi_R^2 \gamma^{(1)}  ] + (1-2\eps) a \big(1-\Tr [\chi_R^2\gamma^{(1)}]\big) \right)\\
&=  a_\eps \Tr [\gamma^{(1)}] + (1-2\eps) a \big(1-\Tr [\gamma^{(1)}]\big)
\end{align*}
Since $a_\eps \le a <0$, we conclude that $\Tr [\gamma^{(1)}] = 1$ as stated in~\eqref{eq:concentration-case} and the proof is complete. 
\end{proof}

\subsection{General case: Proof of Theorem~\ref{thm:general_intro}} \label{sec:general-proof}

We are now able to prove our main result, Theorem \ref{thm:general_intro}, which was stated in the introduction. Our strategy is to split the energy into two parts corresponding to the particles staying in a neighborhood of 0 and those escaping to infinity. We use the weak de Finetti Theorem \ref{thm:weak-De-Finetti} for the local part. The particles far from the origin form a fully translation-invariant system for which we have already shown that Hartree's theory is valid. The conclusion then follows from the binding inequality in Hartree's theory.

\begin{proof} Let $\Psi_N$ be a sequence of wave functions such that $\langle \Psi_N, H_N^V \Psi_N \rangle = E^V(N)+o(N)$ and denote by $\gamma_N^{(k)}$ the $k$-body density matrix of $\Psi_N$. Up to a subsequence we may assume that $\gamma_N^{(k)}\wto \gamma^{(k)}$ weakly-$\ast$ in the trace class for all $k\ge 1$. Let us denote by $\mu$ the probability measure on $B\gH$ associated with $\{\gamma^{(k)}\}_{k=1}^\infty$ as in the weak de Finetti Theorem \ref{thm:weak-De-Finetti}.

First we proceed similarly as before. Let $\chi_R^2 +\eta_R^2=1$ be a smooth partition of unity as in Lemma \ref{lem:splitting} and consider the associated localized states $G_N^\chi$ and $G_N^\eta$ in the Fock space $\cF(\gH)$. By (\ref{eq:split-energy}), we have
\begin{align} \label{eq:split-energy-general}
\lim_{N\to \infty}\frac{E^V(N)}{N} &= \lim_{N\to \infty} \left( \Tr_\gH[T\gamma_N^{(1)}] + \frac{1}{2} \Tr_{\gH^2}[w \gamma_N^{(2)}] \right) \nn\\
&\ge  \liminf_{R\to \infty}\liminf_{N\to \infty} \left\{ \Tr_\gH[T \chi_R \gamma_N^{(1)}\chi_R] + \frac{1}{2} \Tr_{\gH^2}[w \chi_R^{\otimes 2}\gamma_N^{(2)}\chi_R^{\otimes 2}] \right. \nn\\
&\qquad\left. +\Tr_\gH[ K \eta_R \gamma_N^{(1)}\eta_R] + \frac{1}{2} \Tr_{\gH^2}[w\eta_R^{\otimes 2}\gamma_N^{(2)}\eta_R^{\otimes 2}] \right\}.
\end{align}
By the strong local compactness (as in (\ref{eq:wlsc-chi-limit-N})) and the weak quantum de Finetti Theorem~\ref{thm:weak-De-Finetti}, we infer
\begin{multline}
\liminf_{N\to \infty} \left\{ \Tr_\gH[T \chi_R \gamma_N^{(1)}\chi_R] + \frac{1}{2} \Tr_{\gH^2}[w \chi_R^{\otimes 2}\gamma_N^{(2)}\chi_R^{\otimes 2}]\right\}\\
\geq  \Tr_\gH[T \chi_R \gamma^{(1)}\chi_R] + \frac{1}{2} \Tr_{\gH^2}[w \chi_R^{\otimes 2}\gamma^{(2)}\chi_R^{\otimes 2}]= \int_{B\gH} \E^V_{\rm H}(\chi_R u) d\mu(u).
\label{eq:localization-chi-general}
\end{multline}
For the second term of the right side of (\ref{eq:split-energy-general}), by using the geometric  localization method we can show that 
\begin{multline}
\label{eq:localization-eta-general}
\liminf_{N\to \infty} \left( \Tr[T \eta_R \gamma_N^{(1)}\eta_R] + \frac{1}{2} \Tr_{\gH^2}[w \eta_R^{\otimes 2}\gamma_N^{(2)}\eta_R^{\otimes 2}]\right)\\ \ge \int_{B\gH} e^{0}_{\rm H}(1-\norm{\chi_R u}^2) d\mu(u). 
\end{multline}

Before proving (\ref{eq:localization-eta-general}), we explain how to conclude the proof of the theorem. By substituting (\ref{eq:localization-chi-general}) and (\ref{eq:localization-eta-general}) into (\ref{eq:split-energy-general}) and using Fatou's lemma we find that
\bqq
\lim_{N\to \infty}\frac{E^V(N)}{N} &\ge & \liminf_{R\to \infty} \left( \int_{S\gH} \left[\E^V_{\rm H}(\chi_R u)+ e^0_{\rm H}(1-\norm{\chi_R u}^2) \right] d\mu(u)\right) \hfill\\
&\ge & \int_{B\gH} \liminf_{R\to \infty} \left[\E^V_{\rm H}(\chi_R u)+ e^0_{\rm H}(1-\norm{\chi_R u}^2)\right] d\mu(u)\hfill\\
&=& \int_{B\gH} \left[\E^V_{\rm H}(u)+ e^0_{\rm H}(1-\norm{u}^2) \right] d\mu(u)\\
&\geq & \int_{B\gH} \left[\eH^V(\norm{u}^2)+ e^0_{\rm H}(1-\norm{u}^2) \right] d\mu(u)\ge e_{\rm H}(1). 
\eqq
Here we have used the continuity of $\lambda \mapsto e_{\rm H}^0(\lambda)$, which will be proved below, and the non-strict inequality 
$\eH^V(1)\leq \eH^V(\lambda)+\eH^0(1-\lambda)$ which is taken from Lemma~\ref{le:Hartree}.
Given the upper bound $E^V(N)/N\le e^V_{\rm H}(1)$, we conclude that $\lim_{N\to \infty} E(N)/N = e^V_{\rm H}(1)$ and that $\mu$ has its support in 
$$\cM^V=\left\{u\in B\gH\ :\ \cEH^V(u)=\eH^V(\norm{u}^2)=e_{\rm H}^V(1)-e_{\rm H}^0(1-\norm{u}^2)\right\}.$$
Moreover, when $e_{\rm H}^V(1)<e^V_{\rm H}(\lambda)+e_{\rm H}^0(1-\lambda)$ for all $0\le \lambda<1$, we can deduce stronger statements as in the proof of Theorem \ref{thm:confined}. 

Now we prove (\ref{eq:localization-eta-general}). Let us consider the $\eta_R$-localized state $G_N^\eta$ of $\Psi_N$ which is such that 
\begin{align*}
&\Tr_\gH[ K \eta_R \gamma_N^{(1)}\eta_R] + \frac{1}{2} \Tr_{\gH^2}[w \eta_R^{\otimes 2}\gamma_N^{(2)}\eta_R^{\otimes 2}]\nn\\
&\qquad\qquad=\frac{1}{N} \sum_{k=1}^N \Tr_{\gH^k}\left[ \left( \sum_{i=1}^k K_i + \frac{1}{N-1} \sum_{i<j}^k w_{ij}\right) {G}^{\eta}_{N,k} \right] \hfill\\
&\qquad\qquad\ge\sum_{k=1}^N \frac{\Tr {G}^{\eta}_{N,k}}{N} \cdot \inf \sigma_{\gH^k} \left( \sum_{i=1}^k K_i + \frac{1}{N-1} \sum_{i<j}^k w_{ij}\right) . 
\end{align*}
On the other hand, by the fundamental relation $\tr G_{N,k}^\eta=\tr G_{N,N-k}^\chi$ (mentioned before in~\eqref{eq:relation-geometric}) and Theorem \ref{thm:other-localization}, we have 
\begin{multline*}
\lim_{N\to\ii}\sum_{k=0}^N \tr G_{N,k}^\eta\; e_{\rm H}^0 \left(\frac{k}{N}\right) =\lim_{N\to\ii}\sum_{k=0}^N \tr G_{N,N-k}^\chi \;e_{\rm H}^0 \left(\frac{k}{N}\right)\\
=\lim_{N\to\ii}\sum_{k=0}^N \tr G_{N,k}^\chi \;e_{\rm H}^0 \left(1-\frac{k}{N}\right)=\int_{B\gH}\;e_{\rm H}^0 (1-\|\chi_R u\|^2)d\mu(u).
\end{multline*}
Therefore, in order to prove (\ref{eq:localization-eta-general}) it suffices to show that
\begin{equation} \label{eq:localization-eta-bk}
\lim_{N\to\ii}\sum_{k=0}^N \tr G_{N,k}^\eta \left( \frac{k}{N} b_k \left( \frac{k-1}{N-1}\right) - e_{\rm H}^0 \left(\frac{k}{N}\right) \right) = 0,
\end{equation}
where $b_1(\lambda)\equiv 0$ and 
$$b_k(\lambda):= \frac{1}{k}\inf \sigma_{\gH^k} \left( \sum_{i=1}^k K_i + \frac{\lambda}{k-1} \sum_{i<j}^k w_{ij} \right)~~\text{when}~k\ge 2.$$
Note that we have to deal here with a Hamiltonian of the same form as $H^0_N$ but with a factor $\lambda\in[0,1]$ in front of the interaction. By Theorem~\ref{thm:tr-in}, we know that Hartree's theory is correct for such Hamiltonians, that is, we know that 
\begin{align*}
\lim_{k\to\ii} \lambda\, b_k(\lambda)&=\lambda \inf_{\norm{u}^2=1}\left(\pscal{u,Ku}+\frac{\lambda}{2}\iint w(x-y)|u(x)|^2|u(y)|^2\right)\\
&=\inf_{\norm{u}^2=\lambda}\left(\pscal{u,Ku}+\frac{1}{2}\iint w(x-y)|u(x)|^2|u(y)|^2\right)=\eH^0(\lambda).
\end{align*}
So we are almost done. In order to justify~\eqref{eq:localization-eta-bk}, let us show that the functions $\{b_k\}_{k=1}^\infty$ are equicontinuous on $[0,1]$. By adapting the variational estimate (\ref{eq:A-tB}), we obtain 
$$b_k(\lambda)\ge b_k(\lambda')~~\text{for all}~0\le \lambda<\lambda'\le 1.$$
On the other hand, if we denote $\delta := (\lambda'-\lambda)(\alpha^{-1}-\lambda)^{-1}$, then
\bqq
\frac{1}{k}\left( \sum_{i=1}^k K_i + \frac{\lambda'}{k-1} \sum_{i<j}^k w_{ij} \right) &=& \frac{1-\delta}{k} \left( \sum_{i=1}^k K_i + \frac{\lambda}{k-1} \sum_{i<j}^k w_{ij} \right) \hfill\\
&~& + \frac{\delta}{k\alpha} \left( \alpha \sum_{i=1}^k K_i + \frac{1}{k-1}\sum_{i<j}^k w_{ij} \right) \hfill\\
&\ge & (1-\delta) b_k(\lambda) - \frac{C\delta}{\alpha}.
\eqq
Thus
$$ 0\le b_k(\lambda)- b_k(\lambda') \le \delta (b_k(\lambda)+C\alpha^{-1})\le C|\lambda'-\lambda| $$
for a constant $C>0$ independent of $\lambda$, $\lambda'$ and $k$. 

The equicontinuity of $\{b_k\}_{k=1}^\infty$ and the pointwise convergence $\lim_{k\to \infty} \lambda b_k(\lambda)=e_{\rm H}^{0}(\lambda)$ yield the uniform convergence
$$ \lim_{M \to \infty} \sup_{N\ge k\ge M} \left|\frac{k}{N} b_k\left(\frac{k-1}{N-1}\right)-e_{\rm H}^{0}\left(\frac{k}{N}\right)\right| =0.$$
Consequently, we find that
$$ \lim_{N \to \infty} \sup_{k=1,2,...,N} \left|\frac{k}{N} b_k\left(\frac{k-1}{N-1}\right)-e_{\rm H}^{0}\left(\frac{k}{N}\right)\right| =0$$
and (\ref{eq:localization-eta-bk}) follows. The proof is complete.
\end{proof}

\section{Further extensions} \label{sec:extensions}

We conclude this paper by mentioning four interesting cases that may also be dealt with using our method.

\begin{remark}[Bosons in a magnetic field] \label{rmk:magnetic} Our results in Theorem \ref{thm:general_intro} are still valid when the fractional Laplacian $(m^2-\Delta)^{s}$ is replaced by its magnetic version $(m^2+|\nabla+iA(x)|^2)^{s}$, where $A:\R^d \to \R^d$ is a Borel measurable vector potential. For simplicity we assume that 
$|A|^{2s}=f_5+f_6$ with $f_j$ being as in (\ref{eq:assumption-V-w-2}). In this case, the IMS-type estimate 
\begin{equation}
\lim_{R\to \infty} \| K- \chi_R K \chi_R -\eta_R K \eta_R \|_{L^2\to L^2}=0
\label{eq:IMS-magnetic} 
\end{equation}
in (\ref{eq:IMS-estimate}) still holds true with $K$ replaced by $(m^2+|\nabla+iA(x)|^2)^{s}$. The proof of~\eqref{eq:IMS-magnetic} follows the same argument as~\cite[Lemma 7]{LenLew-11}. Moreover, by using the Cauchy-Schwarz inequality and the operator monotonicity of $t \mapsto t^{s}$ when $0<s\le 1$, we can show that 
$$
\lim_{R\to 0}\norm{(1-\Delta)^{-s/2}\eta_R \Big((m^2+|\nabla+iA(x)|^2)^{s} - (m^2-\Delta)^{s} \Big)\eta_R(1-\Delta)^{-s/2}}=0,
$$
which is a substitution for (\ref{eq:control_V_outside}). Therefore, our approach applies exactly as in the non-magnetic case. 
\end{remark}

\begin{remark}[Bosons hoping on a lattice] \label{rmk:lattice}
In this paper we mainly considered continuous systems for simplicity. Our method applies as well to bosons living on a lattice $\mathcal{L}\subset\R^d$ (a discrete subgroup of $\R^d$ with compact fundamental domain) and with a kinetic energy described by the discrete Laplacian. In this case we simply assume that the potentials $V$ and $w$ are in $\ell^\ii(\mathcal{L})$ and tend to zero at infinity (they are then compact operators on $\gH=\ell^2(\mathcal{L})$). As there is always local compactness on the lattice and as the discrete Laplacian satisfies an IMS localization formula similar to~\eqref{eq:IMS-magnetic}, our method applies \emph{mutatis mutandis} and  Theorem~\ref{thm:general_intro} holds in this case as well, without any change.
\end{remark}

\begin{remark}[The absolute ground state]\label{rmk:absolute}
Our method may also be used to investigate the absolute ground state energy of a quantum mechanical Hamiltonian, that is the infimum of the spectrum with no symmetry restriction on the admissible states. The absolute ground state energy coincides with the bosonic ground state energy in the situation covered by Theorem \ref{thm:general_intro} by a well-known method \cite[Section 3.24]{LieSei-09}, but it need not be the case in general, for example in the presence of magnetic fields or rotation (see Example \ref{ex:magn rot} and Remark \ref{rmk:magnetic}).

Observe that any Hamiltonian $H_N$ of the form \eqref{eq:intro hamil} satisfies 
\[
U_{\sigma_N} H_N U_{\sigma_N} ^*  = H_N
\]
for all permutations $\sigma_N$ of $N$ variables, where $U_{\sigma_N}$ is the unitary operator permuting the order of variables according to $\sigma_N$. Therefore, when analyzing the absolute ground state energy of $H_N$ we may consider only the mixed symmetric states, which are positive trace-class operators $\gamma ^{(N)}$ acting on  $\gH ^{\otimes N}$ and satisfying
\begin{equation}\label{eq:symmetry}
U_{\sigma_N} \gamma ^{(N)} U_{\sigma_N} ^* =\gamma ^{(N)}
\end{equation}
for all permutations $\sigma_N$. In this language the assumption on the Bose-Einstein symmetry corresponds to the stronger condition
\begin{equation}\label{eq:BE symmetry}
U_{\sigma_N} \gamma ^{(N)} =  \gamma^{(N)} U_{\sigma_N} =  \gamma^{(N)}
\end{equation}
for all permutations $\sigma_N$.

If we are given an infinite sequence $\left\{\gamma ^{(k)}\right\}_{k=0} ^{\infty}$ of $k$-particle positive trace-class operators that satisfy the symmetry assumption (\ref{eq:symmetry}) and the consistency assumption \eqref{eq:consistent}, then a generalization of Theorem 2.1 proved in \cite{Stormer-69,HudMoo-75} implies that there exists a Borel probability measure $\mu$ on the set of positive trace class operators on $\gH$ with trace $1$, 
such that
\begin{equation}
\gamma^{(k)}=\int \gamma^{\otimes k} \, d\mu(\gamma)
\label{eq:melange3}
\end{equation}
for all $k\geq0$. 

 It is then not difficult to adapt our approach to deduce a weak version as in Section~\ref{sec:de_Finetti} and use it to prove results about the absolute ground state that parallel those we presented for the bosonic ground state.

In this case, one obtains as limit object a Hartree theory for mixed one-body states, that is the ground state energy per particle is given in the limit $N\to \infty$ by the minimization of the functional
\begin{equation}\label{eq:Hartree mixed}
\cEH (\gamma):= \tr_{\gH} [T\gamma] + \frac12\tr_{\gH ^{\otimes 2}} [w\gamma ^{\otimes 2}]
\end{equation}
over all positive trace-class operators $\gamma$ on $\gH$ with $\Tr \gamma=1$. 

The general question of when the minimization of \eqref{eq:Hartree mixed} reduces to that of \eqref{eq:intro Hartree} (that is when absolute minimizers are asymptotically bosonic in the mean-field limit) seems to be mostly open. In \cite{Sei-03} the absolute ground state energy of a rotating trapped Bose gas with repulsive interactions has been considered in the Gross-Pitaevskii limit, and a functional similar to \eqref{eq:Hartree mixed} has been derived. In this particular case, sufficient and necessary conditions are also given for the minimization of \eqref{eq:Hartree mixed} to reduce to that of \eqref{eq:intro Hartree}. These conditions are intimately linked to the question of symmetry breaking and nucleation of vortices in rotating Bose gases. 
\end{remark}


\appendix

\section{Alternative proof of the weak de Finetti theorem}\label{app:Hudson-Moody-proof}

Here we explain how to derive Theorem \ref{thm:weak-De-Finetti} from the strong de Finetti theorem, following ideas of Hudson and Moody~\cite{HudMoo-75}. This proof is more direct but not as constructive as the one we gave in Section~\ref{sec:de_Finetti}, which also allowed us to relate the de Finetti measure $\mu$ to geometric localization in Theorem \ref{thm:other-localization}.

Let us consider a sequence of normal states $\Gamma_N \in \gS ^1 (\gH^N)$ and the weak-$\ast$ limits $\gamma^{(k)}$, $k\geq 1$, of their reduced density matrices, as in the statement of Theorem \ref{thm:weak-De-Finetti}. Following~\cite{RagWer-89,Werner-92}, it is useful to think of $\Gamma_N$ as a state $\omega_N$ on $\cB (\gH ^{\otimes\infty})$ which is the $C^{\ast}$ inductive limit of the sequence $\cB(\gH^{\otimes N})$ of all bounded operators on the tensor product $\gH^{\otimes N}$ without symmetry. This means that
$$\omega_N (b_1\otimes \ldots \otimes b_M) := \tr_{\gH^{\otimes N}} \left(\Gamma_N\ b_{1}\otimes \ldots \otimes b_{N}\right)$$
for all $M\geq N$ and where $\Gamma_N$ is extended to 0 outside of $\gH^N\subset\gH^{\otimes N}$. Here we have made the abuse of notation to identify the operator $b_1\otimes \ldots \otimes b_M$ with $b_1\otimes \ldots \otimes b_M\otimes \1\otimes\ldots$. Note that, using the bosonic symmetry of $\Gamma_N$
\begin{equation}\label{eq:symmetry finite N}
\omega_N (b_1\otimes \ldots \otimes b_M)  =  \omega_N (b_{\sigma(1)}\otimes \ldots \otimes b_{\sigma(N)}\otimes b_{N+1} \otimes \ldots \otimes b_M)
\end{equation}
for any $M\geq N$ and any permutation $\sigma$ of the first $N$ variables.

By the Banach-Alaoglu Theorem, the sequence $\omega_N$ admits a weak-$\ast$ cluster point $\omega$, a state on $\cB(\gH^{\otimes\ii})$. So $\omega_N$ converges to $\omega$ along a subnet, which means that for any $n$ and any $b_1,\ldots,b_n \in \cB (\gH)$
\begin{equation}\label{eq:extract abstract}
\omega_{h(\alpha)} (b_1 \otimes \ldots \otimes b_n) \to \omega (b_1\otimes \ldots \otimes b_n) 
\end{equation}
where $h:A\mapsto \mathbb{N}$ is a monotone cofinal function from some directed set $A$ to the integers. It is of course important to be able to test against the identity operator in \eqref{eq:extract abstract}, to ensure that $\omega$ is a state.

The state $\omega$ determines a hierarchy of $n$-particle states $\omega^{(n)}$ on $\cB(\gH^{\otimes n})$ with the consistency relations
\begin{equation}\label{eq:consistent abstract}
\omega^{(m)} (b_1\otimes \ldots \otimes b_m) =\omega (b_1\otimes \ldots \otimes b_m )= \omega^{(n)} (b_1\otimes \ldots \otimes b_m \otimes \1 ^{\otimes n-m} )
\end{equation}
for all $n\geq m$, and all $b_1, \ldots, b_m \in \cB(\gH)$. 

From~\eqref{eq:symmetry finite N} and~\eqref{eq:extract abstract} we deduce that the cluster point $\omega$ is \emph{symmetric}, that is, $\omega^{(n)}(b_{\sigma(1)}\otimes\cdots\otimes b_{\sigma(n)})=\omega^{(n)}(b_1\otimes\cdots\otimes b_n)$ for all bounded operators $b_1,...,b_n\in\cB(\gH)$ and every permutation $\sigma$. Now, from the strong de Finetti theorem (for abstract states on an algebra) of \cite{Stormer-69,HudMoo-75}, there exists a Borel probability measure $\mu$ on the set of states $\cS (\cB(\gH))$ on $\cB(\gH)$ such that, for any $n\geq 0$,
\begin{equation}\label{eq:de Finetti C star}
\omega ^{(n)} = \int_{\cS (\cB(\gH))} d\mu (\omega)\, \omega ^{\otimes n}.
\end{equation}
This is a consequence of the fact, proved first in \cite{Stormer-69}, that the tensor powers $\omega ^{\otimes \infty}$ are the extreme points of the convex set of symmetric states on $\cB(\gH ^{\infty})$. The link between $\omega^{(n)}$ and the weak-$\ast$ limits $\gamma^{(n)}$ of the density matrices of the sequence $\Gamma_N$ is that 
\begin{equation}
\tr_{\gH^N}\big(\gamma^{(n)} K\big) = \lim_{N\to\ii}\tr_{\gH^N}\big(\gamma_N^{(n)} K\big)= \lim_{N\to\ii} \omega_N(K)=\omega(K)=\omega^{(n)}(K)
\label{eq:coincide-compact} 
\end{equation}
for every symmetric compact operator $K$ on $\gH^N$.

We recall that any state $\omega$ on $\cB(\mathfrak{K})$ can be restricted to the algebra $\cK(\mathfrak{K})$ of compact operators (here $\mathfrak{K}$ is  any fixed separable Hilbert space). The associated non-negative linear form $\omega _{\rm nor}$ is called the \emph{normal part} of $\omega$ and it necessarily arises from a trace-class density operator $\gamma_\omega\in\gS^1(\mathfrak{K})$, since $\gS^1(\mathfrak{K})$ is the dual of $\cK(\mathfrak{K})$. Indeed, the map $\omega\mapsto \gamma_\omega$ is continuous from the set of states into the trace-class and we have 
$$\tr(\gamma_\omega)=\sup_{\substack{K\in\cK(\mathfrak{K})\\ 0\leq K\le 1}}\tr(\gamma_\omega K)=\sup_{\substack{K\in\cK(\mathfrak{K})\\ 0\leq K\le 1}}\omega(K)\leq \sup_{\substack{B\in\cB(\mathfrak{K})\\ 0\leq B\le 1}}\omega(B)= \omega(\1)=1.$$
By~\eqref{eq:coincide-compact}, we deduce that $\gamma^{(n)}$ must be the density operator associated with the normal part of the state $\omega^{(n)}$.

The main point of the proof is the remark that the normal part of a tensor product is the tensor product of the normal part,
\begin{equation}\label{eq:normal power}
\left( \omega ^{\otimes n}\right) _{\rm nor} = \left( \omega _{\rm nor} \right) ^{\otimes n},
\end{equation}
which follows from the fact that $\cK(\gH^{\otimes N})=\cK(\gH)^{\otimes N}$. From this we deduce that 
\begin{equation}\label{eq:de Finetti normal}
\gamma^{(n)} = \int_{\cS (\cB(\gH))} d\mu (\omega) (\gamma_\omega) ^{\otimes n}. 
\end{equation} 
Since, as we have said above, the map $\omega\mapsto \gamma_\omega$ is continuous, we can consider the push-forward Borel probability measure $\tilde\mu$ on the unit ball $B\gS^1(\gH)$ of the trace-class, which is such that $\tilde\mu(A)=\mu(\{\omega\, :\, \gamma_\omega\in A\})$ for all Borel sets $A\subset B\gS^1(\gH)$. Hence
$$\gamma^{(n)} = \int_{B\gS^1(\gH)} d\tilde\mu (\gamma)\, \gamma^{\otimes n}.$$

To conclude the proof of the weak de Finetti theorem, there only remains to show that $\tilde\mu$ is supported on the set of pure states:
$\tilde\mu \left( \left\{ \ketl u \ketr \bral u \brar, \: u \in B\gH  \right\}\right) = 1$. 
This follows exactly \cite[Section 4]{HudMoo-75} and this is where we need the important fact that our state has the Bose-Einstein symmetry, that is, $S_n \gamma ^{(n)} = \gamma ^{(n)}$ where $S_n$ is the symmetrization operator. Taking the trace against $S_n$ we find
\begin{align*}
\tr\gamma^{(n)} =\tr S_n\gamma^{(n)}&= \int_{B\gS^1(\gH)} d\tilde\mu (\gamma)\, \tr \big(S_n \gamma^{\otimes n}\big)\\
&\leq \int_{B\gS^1(\gH)} d\tilde\mu (\gamma)\, \tr \gamma^{\otimes n}=\tr\gamma^{(n)}
\end{align*}
where we have used that $S_n\leq1$. From this we deduce that $\tr (S_n\gamma^{\otimes n})=\tr \gamma^{\otimes n}$ for all $n\geq1$ and $\tilde\mu$-almost all $\gamma$. This is equivalent to $\gamma=|u\rangle\langle u|$ by~\cite[Proposition 3]{HudMoo-75}. Therefore $\tilde\mu$ is supported on rank-one density operators. Associated with this measure, there is a unique $S^1$-invariant Borel probability measure $\mu'$ on the ball $B\gH$ such that 
$$\gamma^{(n)} = \int_{B\gH} d\mu' (u) |u^{\otimes n}\rangle\langle u^{\otimes n}|,$$ 
and $\mu'$ is the sought-after de Finetti measure of Theorem~\ref{thm:weak-De-Finetti}.\qed



\end{document}